\documentclass[journal]{IEEEtran}

\IEEEoverridecommandlockouts

\usepackage{bm}
\usepackage{epsf}
\usepackage{subfigure}
\usepackage{cite}
\usepackage{graphics} 
\usepackage{epsfig} 
\usepackage{amsthm}
\usepackage{mathrsfs}
\usepackage{mathptmx} 
\usepackage{times} 
\usepackage{amsmath} 
\usepackage{amssymb}  
\usepackage{psfrag}
\usepackage{enumerate}
\usepackage{url}
\usepackage{stfloats}
\usepackage{array}
\usepackage{latexsym}
\usepackage{amssymb} 
\usepackage[usenames]{color}
\usepackage[ruled,vlined]{algorithm2e}
\usepackage{wrapfig}

\newtheorem{definition}{Definition}
\newtheorem{theorem}{Theorem}

\newtheorem{proposition}{Proposition}

\newtheorem{corollary}{Corollary}
\newtheorem{remark}{Remark}
\newtheorem{example}{Example}
\newtheorem{claim}{Claim}

\newcommand{\lfteqn}{\begin{eqnarray} \begin{array}{lllllll}}
\newcommand{\ndeqn}{\end{array} \nonumber \end{eqnarray}}
\newcommand{\Lfteqn}{\begin{eqnarray} \begin{array}{lllllll}}
\newcommand{\Ndeqn}{\end{array}  \end{eqnarray}}

\begin{document}

\title{\LARGE \bf Online State Estimation for Supervisor Synthesis in Discrete-Event Systems with Communication Delays and Losses
}

\author{Yunfeng~Hou,  ~Yunfeng~Ji,~\IEEEmembership{Member,~IEEE},  ~Gang~Wang, ~\IEEEmembership{Member,~IEEE}, ~Ching-Yen~Weng,  and  ~Qingdu~Li

\thanks{Yunfeng Hou, Yunfeng Ji, Gang Wang, and Qingdu Li are with the Institute of Machine Intelligence, University of Shanghai for
Science and Technology, Shanghai 200093, China. E-mail: yunfenghou@usst.edu.cn, ji$\_$yunfeng@usst.edu.cn, 2010wanggang@gmail.com, and liqd@usst.edu.cn.

Ching-Yen Weng ({weng0025@e.ntu.edu.sg}) is with the Robotics Research Centre, Nanyang Technological University, Singapore.
}}


\maketitle


\begin{abstract}
In the context of networked discrete-event systems (DESs), communication delays and losses exist between the plant and the supervisor for
observation and between the supervisor and the actuator for control.
In this paper, we first introduce a new framework for supervisory control of networked DESs.
Under the introduced framework, we address the state estimation problem for supervisor synthesis of networked DESs with both communication delays and losses.
The estimation algorithm considers the effect of the controls imposed on the system.
Additionally, the estimation algorithm is based on the control decisions available up to the moment, and all the future control decisions are assumed to be unknowable.
Two notions, called ``observation channel configuration'' for tracking observation delays and losses and ``control channel configuration'' for tracking control delays and losses, are defined.
Then, we introduce an online approach for state estimation of the controlled system.
Compared with the existing approach, the proposed approach under the introduced framework can estimate the state of the controlled system more accurately.
As an application of  the proposed approach, we finally show that the existing methods can be easily applied to synthesize maximally permissible and safe networked supervisors.
\end{abstract}


\section{Introduction}

\IEEEPARstart{T}{he} \textcolor{blue}{dynamics of DESs is driven by asynchronous event sequences. 
A state estimate of DESs is defined as the set of  (discrete) states such that the controlled system or the closed-loop system may be in after observing a sequence of observable events.
{Supervisor synthesis} is an important problem in the supervisory control of DESs  and has drawn much attention in the DES community  over the past decades \cite{ramadge1987supervisory,lin88is,heymann94deds,1996Centralized,2002Effective,2015Relative,yin16tac1,yin16tac2, yin17tac, yin18tac,alves17tac, ji21tac}.
In {supervisor synthesis}, a supervisor or controller is desired to dynamically enable or disable event occurrences so that state estimates of the controlled system  always satisfy the property of interest \cite{heymann94deds, yin16tac1,yin16tac2, yin17tac, yin18tac}.
Thus, state estimation is crucial in supervisor synthesis for determining a valid control action after each new event observation.}


Recent advances in WLAN-based and cellular-based communication systems have enabled us to connect a supervisor and a plant via a shared communication network (networked DESs).
Such a networked structure provides more flexible and agile ways to control a DES.
For example, it allows the supervisor to be an edge computing node, which means it can share the computing resource with other edge computing nodes.
However, communication delays and losses existing in networks pose significant challenges to attaining accurate state estimations  when trying to solve the {supervisor synthesis problem}.
Most of the current works implement state estimations of networked DESs based on the open-loop system without using the information of the control actions executed by the actuator \cite{wang15cdc, shu15tac,shu17tac1,alves19acc, hou20tac,lin20tcns,zhou21tase,shu17tac2,lin2014control}.
When the state estimates calculated in \cite{wang15cdc, shu15tac,shu17tac1,alves19acc, hou20tac,lin20tcns, zhou21tase,shu17tac2,lin2014control} are used for supervisor synthesis, they may contain some states that have been prevented from being reached by the control.
One exception is the work of \cite{liu21tac}, where a novel online state estimation algorithm  was proposed by taking the information of the control decision's history into consideration.
{Nevertheless, as will be shown in Section V, the state estimate calculated in \cite{liu21tac} is an overapproximation of the actual state estimate of the controlled system and may contain states that the controlled system  never reaches.
Additionally, this approach considers only control delays.}

In this paper, we generalize the assumption of control delays in \cite{liu21tac} into a more general case of communication delays and losses and study how to obtain an accurate state estimate of the controlled system under the observation delays and losses and the control delays and losses.
Specifically, this paper assumes that:  (1) both the control channel and the observation channel satisfy first-in-first-out (FIFO);
(2) the observation delays and the control delays are upper bounded by $N_o$ and $N_c$ event occurrences (observable or not), respectively;
(3) the numbers of consecutive observation losses  and  consecutive control losses are no larger than $N_{l,o}$ and $N_{l,c}$, respectively.
\textcolor{blue}{It is worth noting that when there exists only control delays and losses, the observation of the supervisor to a string is deterministic, and we can immediately determine which control decision  has been made after the occurrence of this string.
However, when observation delays and losses exist, the observation of the supervisor to a string is nondeterministic and varies with the different observation delays and losses.
For different observations, the supervisor may make different control decisions.
An event after this string may be allowed to occur for some of these control decisions, but not be allowed to occur for the other control decisions.
We must consider all the possibilities, which complicate the state estimation problem.}

To obtain an accurate state estimate of the controlled system, we must first specify an accurate ``dynamics'' of the controlled system.
To this end, we introduce a new framework for supervisory control of networked DESs.
In this framework, a communication automaton is constructed to model the interaction process between the supervisor and the plant over the control channel and the observation channel under communication delays and losses.
Each state of the communication automaton records: (i)
the state that the plant is in, (ii) the state that the networked supervisor is in, (iii) the sequence of observable events that have occurred but still need to be delivered to the supervisor, (iv) the number of consecutive observation losses, (v) the control action in use,
(vi) the sequence of control actions that have been issued but are still delayed at the control channel, and (vii) the number of consecutive control losses.
States of the communication automaton are updated when one of the following behaviors occurs (represented as a special event occurrence): (a) a new event occurs, (b) a new observable event is communicated, (c) a new control action is executed, (d) an observation loss occurs, and (e) a control loss occurs.
The dynamics of the controlled system can be simply ``decoded'' from sequences that can be generated by the communication automaton.

Next, we discuss how to produce online state estimates of the controlled system under communication delays and losses.
Specifically, for tracking states of the observation channel, we introduce the  ``observation channel configuration'', which consists of two parts: (i) a sequence of event-integer pairs that is used to track the delayed observable event occurrences and the number of observation delays, and (ii) an integer that is used to track the number of consecutive observation losses.
On the other hand, for tracking the states of the control channel, we introduce the ``control channel configuration'', which consists of three parts: (a) an admissible control action that is taking effect, (b) a sequence of command-integer pairs that is used to track the delayed control actions and the number of control delays, and (c) an integer that is used to track the number of consecutive control losses.
By incorporating the control channel configurations and the observation channel configurations into the states of the plant, we can obtain a triplet.
We call such a triplet an augmented state (the plant state is augmented with the observation channel configuration and the control channel configuration).
An online approach is proposed for updating the augmented state estimates upon each new observation, which can be used to estimate the states of the controlled system.
Compared with \cite{liu21tac}, we show that the proposed approach can estimate states of the controlled system more accurately, even if there are only control delays.


\textcolor{blue}{Based on a game structure called \textsl{Bipartite Transition System} (BTS) \cite{yin16tac1}, a general approach for solving a set of important supervisor synthesis problems was proposed by Yin and Lafortune in \cite{yin16tac1,yin16tac2}.
Benefiting from the state estimation algorithm developed in this paper, it becomes a reality that a BTS can be extended to its networked counterpart \textsl{Networked Bipartite Transition System} (NBTS) when communication delays and losses exist.
Using the NBTS,  techniques developed in \cite{yin16tac1,yin16tac2} can be easily extended to solve the corresponding problems with the communication delays and losses, which can effectively simplify the research.
As an example, we finally show how to construct an NBTS using the proposed state estimation approach and extend the techniques developed in \cite{yin16tac1,yin16tac2} to synthesize a maximally permissible networked supervisor while ensuring that the safety of the controlled system is satisfied.}

The rest of this paper is organized as follows.
Section II presents some preliminary concepts and the required assumptions in this paper. Section III specifies the language that may be generated by the controlled system. 
An online procedure for estimating states of the controlled system is presented in Section IV.
Section V compares the proposed approach with the previous approach.
Section VI shows the application of the proposed approach. Section VII concludes this paper.

Due to space limitations, some proofs are omitted in this paper and can be found in the Appendix.

\section{Preliminaries}

\subsection{Preliminaries}
A DES is modeled by a deterministic finite-state automaton 
$
G=(Q, \Sigma, \delta, q_{0}),
$
where $Q$ is the finite set of states; $\Sigma$ is the finite set of events; $\delta: Q\times \Sigma \rightarrow Q$ is the transition function; and $q_{0}$ is the initial state.
$\delta$ is extended to $Q \times \Sigma^{\ast}$ in the usual way.
``$!$'' means ``is defined'', and ``$\not !$'' means ``is not defined''.
$\mathscr{L}(G)$ is the language generated by $G$. 
$\Sigma$ is partitioned into the set of controllable events
$\Sigma_c$ and the set of uncontrollable events $\Sigma_{uc}$.
$\Sigma$ is also partitioned into the set of observable events
$\Sigma_{o}$ and the set of unobservable events $\Sigma_{uo}$. 
The natural projection $P: \mathscr{L}(G) \rightarrow \Sigma_{o}^{\ast}$ is defined as $P(\varepsilon)=\varepsilon$ and, for all $s, s\sigma \in\mathscr{L}(G)$, $P(s\sigma)=P(s)\sigma$ if $\sigma \in \Sigma_{o}$, and  $P(s\sigma)=P(s)$ otherwise.
Given automata $G_{1}$ and $G_{2}$, we say that $G_1$ is a subautomaton of $G_2$, denoted by $G_1 \sqsubseteq G_2$, if $G_1$ and $G_2$ have the same initial state and $G_1$ is a subgraph of $G_2$. 


Given a $s$, let $\overline{\{s\}} = \{s': (\exists s'') s = s's''\}$ be the set of all prefixes of $s$. 
The length of a string $s$ is denoted by $|s|$.
The prefix closure of a language $L\subseteq \Sigma^{\ast}$ is denoted as $\overline{L}$. $L$ is prefix-closed if $L=\overline{L}$. 
We only consider prefix-closed languages in this paper. 
$\varepsilon$ denotes the empty string. 
Given a $s=\sigma_1\sigma_2\cdots\sigma_k$, we write $s^i = \sigma_1 \cdots \sigma_i$ for $i = 1,\ldots,k$, and $s^0 = \varepsilon$.
$\mathbb{N}$ is the set of natural numbers.
Given $a,b\in \mathbb{N}$, let $[a,b]$ be the set of natural numbers between $a$ and $b$.
The cardinality of a set  $Z$ is denoted by $|Z|$. 
Given a $n \in \mathbb{N}$, let $Z^{\le n}$ be the set of sequences (consisting of elements in $Z$) with a length no larger than $n$.


{We consider a networked DES in this paper.}  
Due to the network characteristics, communication delays and losses exist for both
control and observation. We make the following assumptions on the networked DESs:
(i) Both the control channel and the observation channel satisfy the FIFO property, i.e., the observable event occurrences are delivered to the supervisor in the same order as they were generated, and the control actions are executed by the actuator of the plant in the same order as they were issued.
(ii) The communication delays in the observation (regarded as observation delays) are nondeterministic but are upper bounded by $N_o$ events, i.e., when an event occurs, the system can generate no more than $N_o$ event occurrences before this event is communicated to the supervisor. The communication delays in the control (regarded as control delays) are also nondeterministic but are upper bounded by $N_c$ events, i.e., before an issued control action is executed, the system can generate no more than $N_c$ event occurrences;
(iii) The consecutive losses of the observable event occurrences are assumed to be no larger than $N_{l,o}$, i.e., before a new observable event is communicated (observed), there are at most $N_{l,o}$ consecutive observation losses, and the consecutive losses of the control actions are assumed to be no larger than $N_{l,c}$, i.e., before a new control action is executed, there are at most $N_{l,c}$ consecutive control losses.
(iv) The actuators always implement the most recently received action, and the initial control action can be executed without any delays and losses.

\begin{remark}
We assume that both the control channel and the observation channel satisfy FIFO, since it is often the case that there is only one communication channel from the plant to the supervisor and from the supervisor to the plant.
\textcolor{blue}{This is slightly different from \cite{rashidinejad18wodes,alves20tac,tai22ac,lin22ac,zhu22icca}, where FIFO is not required for the communication channels.}
In addition, since the communication losses are usually small, we assume that both the consecutive observation losses and the consecutive control losses have upper bounds.
The same assumption can be found in \cite{zhu19cdc,liu21cdc}.
Meanwhile, the delays are measured by the number of event occurrences (observable or not), which is different from \cite{rashidinejad18wodes,alves20tac,zhao17tcns,park08ac,miao19cdc,tai22ac} where time is explicitly considered. 
\textcolor{blue}{We also assume that the initial control action has been deployed in the execution module of the plant before it starts to work.
Thus, when the plant is initialized, the initial control action can be executed without any delays and losses.}
\end{remark}

As in \cite{alves20tac}, the networked supervisor is defined as a pair $S=(A,\gamma)$ such that $A=(X,\Sigma_o,\xi, x_0)$\footnote{Here, we assume that $X$ is a finite set because all the networked supervisors considered in this paper can always be represented by an automaton with finite state spaces.} is a deterministic automaton with $\mathcal{L}(A)=\Sigma_o^*$, and $\gamma: X \rightarrow 2^\Sigma$ is a function that specifies the set of events to be enabled.
Specifically, for any $t \in \Sigma_o^*$, we denote $\gamma(\xi(x_0,t))$ by the set of events to be enabled after observing $t$.
With a slight abuse of notation, we also write $\gamma(\xi(x_0,t))=S(t)$.
Let $\Pi=\{\pi\in 2^{\Sigma}:\Sigma_{uc}\subseteq \pi\}$ be the set of all admissible control actions.
Since we cannot disable an uncontrollable event, $S(t)$ should be  admissible, i.e., $S(t)\in \Pi$.
\begin{remark}
\textcolor{blue}{Note that when communication losses exist, strings that may be observed by the networked supervisor are no longer $P(\mathcal{L}(G))$ because some observable event occurrences may be lost during the transmission.
To implement supervisory control under observation delays and losses, techniques were developed in \cite{alves20tac} to construct an untimed automaton that models all the possible system observations in the presence of the observation delays and losses. 
And a networked supervisor that maps each possible system observation to an admissible control action was proposed in \cite{alves20tac}.
In contrast to \cite{alves20tac}, the networked supervisor $S$ is defined over the entire $\Sigma_o^*$ in this paper.
For those $t \in \Sigma^*_o$ that can never be observed by  $S$, we can define a special state $x_{spe}$ in $A$ such that $\xi(x_0,t)=x_{spe}$ and $\gamma(x_{spe})=\Sigma_{uc}$.}
\end{remark}

\begin{figure}
\centering \subfigure[$G$]{\label{Fig21}\includegraphics[width=3.4cm]{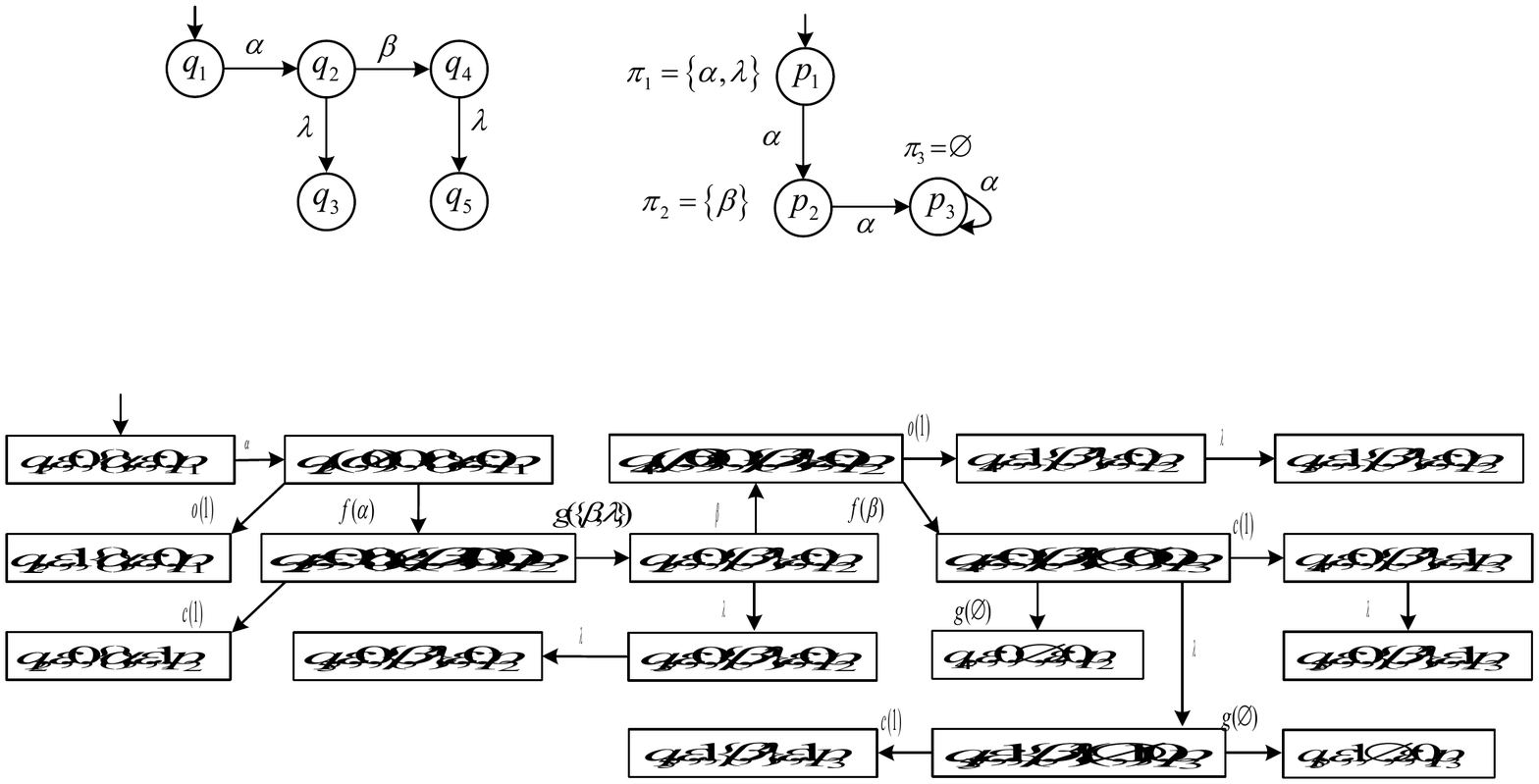}}
\subfigure[$S=(A,\gamma)$]{\label{Fig22}\includegraphics[width=4.0cm]{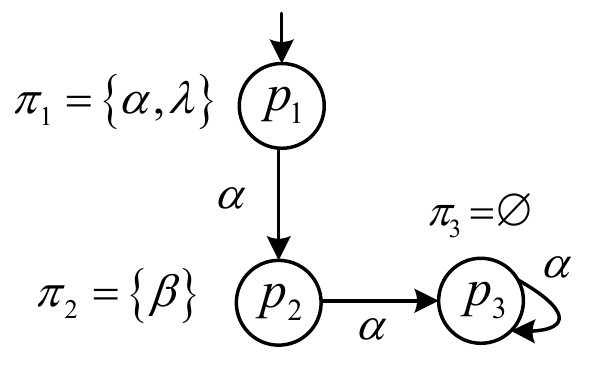}}
\caption{System $G$ and networked supervisor $S=(A,\gamma)$  in Example \ref{Ex1}.} \label{Fig2}
\end{figure}
\begin{example}\label{Ex1}
Consider system $G$ depicted in Fig. \ref{Fig21} with $\Sigma_c=\Sigma=\{\alpha,\beta,\lambda\}$ and  $\Sigma_o=\{\alpha\}$.
The networked supervisor $S=(A,\gamma)$ is depicted in  Fig.\ref{Fig22}.
The function $\gamma$ is specified by the set of events associated with each
state in Fig. \ref{Fig22}.
As shown in Fig. \ref{Fig22}, when $S$ is in $p_1$, $S(\varepsilon)=\gamma(p_1)=\pi_1=\{\alpha,\lambda\}$.
Once $\alpha$ is observed, $S$ moves to state $p_2$, and $S(\alpha)=\gamma(p_2)=\pi_2=\{\beta\}$.
For all $t \in \Sigma^*_o \setminus \{\varepsilon,\alpha\}$,  $S(t)=\gamma(p_3)=\pi_3=\emptyset$.
\end{example}

Given two networked supervisors $S_1$ and $S_2$, we say that $S_1$ is smaller than $S_2$, denoted by $S_1 \subseteq S_2$, if $S_1(t)\subseteq S_2(t)$ for all $t\in \Sigma_o^*$; we say $S_1$ is strictly smaller than $S_2$, denoted by $S_1 \subset S_2$, if $S_1 \subseteq S_2$ and there exists $t \in \Sigma_o^*$ such that $S_1(t)\subset S_2(t)$.

\section{Dynamics of the controlled system}

To accurately estimate the states of a controlled system, we must first specify an accurate language that can be generated by the controlled system.
To this end, we consider a framework of supervisory control of networked DESs in this section.
In this framework, we  construct a communication automaton that explicitly models the interaction process between the plant and the supervisor over the observation channel and the control channel, where communication delays and losses exist.
We show that the ``dynamics'' of the controlled system can be inferred from the constructed communication automaton.



We first define four special types of events to characterize the behaviors of the communication delays and losses.
\begin{enumerate} 

\item To describe the loss of an observable event occurrence, define bijection $o:[1,N_o+1] \rightarrow \Sigma^o$ such that  $\Sigma^o=\{o(i):i\in [1,N_o+1]\}$, where $o(i)$ indicates the loss of the $i$th observable event in the observation channel.

\item To describe the loss of  a control action, define bijection $c:[1,N+1] \rightarrow \Sigma^c$ such that $\Sigma^{c}=\{c(i):i\in [1,N+1]\}$, where  $N=N_c+N_o$ and $c(i)$ indicates the loss of the $i$th control action in the control channel.

\item To keep track of what observable event has been communicated, define bijection $f:\Sigma_o \rightarrow \Sigma^f$ such that $\Sigma^f=\{f(\sigma):\sigma \in \Sigma_o\}$, where $f(\sigma)$ indicates that the occurrence of $\sigma$ has been communicated to the supervisor.


\item To model which control decision is taken, define bijection $g:\Pi \rightarrow \Sigma^g$ such that $\Sigma^g=\{g(\pi): \pi \in \Pi\}$, where $g(\pi)$ indicates that the control action $\pi$ has been executed.
 
\end{enumerate}
Note that  $\Sigma$, $\Sigma^o$, $\Sigma^c$, $\Sigma^f$, and $\Sigma^g$ are mutually disjoint.

Let $N=N_c+N_o$.
Given a system $G$ and a networked supervisor $S=(A,\gamma)$ such that $A=(X,\Sigma_o,\xi,x_0)$, we denote each state of the communication automaton by a seven-tuple $\tilde{q}=(q,x,n,\phi,y,m,p) \in  Q \times (\Sigma_o \times [0,N_o])^{\le N_o+1}\times [0,N_{l,o}] \times \Pi \times (\Pi \times [0,N_c])^{\le N+1}  \times [0,N_{l,c}] \times X$: where (i) $q \in Q$ tracks the state that the plant is in; (ii) $x$ is a sequence of pairs $(\sigma_1,n_1) \cdots (\sigma_k, n_k) \in (\Sigma_o \times [0,N_o])^{\le N_o+1}$ such that $\sigma_1\cdots\sigma_k \in \Sigma_o^*$ tracks a sequence of observable events that have occurred but still need to be communicated (delivered) to the supervisor, and the integer $n_i$ tracks the number of event occurrences while $\sigma_i$ is waiting to be communicated;
(iii) $n \in [0,N_{l,o}]$  counts the number of consecutive observation losses;
(iv) $\phi \in \Pi$ is the control action in use;
(v) $y$ is a sequence of pairs $(\pi_1,m_1) \cdots (\pi_h,m_h) \in (\Pi \times [0,N_c])^{\le N+1}$ such that  $\pi_1 \cdots \pi_h \in \Pi^{\le N+1}$ tracks a sequence of admissible control actions that have been issued but have  not been executed due to control delays, and the integer $m_i$ tracks the number of event occurrences while the control action $\pi_i$ is delayed at the control channel;
(vi) $m \in [0,N_{l,c}]$  counts the number of consecutive control losses;
(vii) $p \in X$ tracks the state that networked supervisor $S$ is in.

\begin{remark}
Note that the lengths of $x$ and $y$ are both finite.
Since the observation delays are assumed to be upper bounded by $N_o$, there could be $N_o$ additional event occurrences at most before an observable event is communicated.
Thus, the number of events delayed at the observation channel is $N_o+1$ at most, and the length of $x$ is no larger than $N_o+1$.
On the other hand, due to control delays and observation delays, the control action in use could be anyone issued in the past $N$ steps. 
When a new event occurs, at least one control action issued in the past $N+1$ steps is executed.
Therefore, the length of $y$ is no longer than $N+1$.
\end{remark}

Given a $x \in (\Sigma_o \times [0,N_o])^{\le N_o+1}$, if $x=(\sigma_1,n_1) \cdots (\sigma_k, n_k) \neq \varepsilon$, we define $\textbf{NUM}(x)=n_1$ as the integer in the first pair of $x$, and if $x=\varepsilon$, we define $\textbf{NUM}(x)=0$.
Since $\sigma_1$ is the first event queued at the observation channel, $\textbf{NUM}(x)$ records the maximum observation delays at the moment.
To update the observation delays after a new event occurrence, we define $x^+=(\sigma_1,n_1+1) \cdots (\sigma_k,n_k+1)$ if $x=(\sigma_1,n_1) \cdots (\sigma_k, n_k) \neq \varepsilon$, and $x^+=\varepsilon$ if $x=\varepsilon$.
Similarly, for any $y \in (\Pi \times [0,N_c])^{\le N+1}$, if $y=(\pi_1,m_1) \cdots (\pi_h, m_h) \neq \varepsilon$, we define $\textbf{NUM}(y)=m_1$ and  $y^+=(\pi_1,m_1+1) \cdots (\pi_h, m_h+1)$, and if $y=\varepsilon$, we define $\textbf{NUM}(y)=0$ and  $y^+=\varepsilon$.






With the above preparations, we formally construct the communication automaton $G_S=(\tilde{Q},\tilde{\Sigma}, \tilde{\delta}, \tilde{q}_{0})$, where $\tilde{Q} \subseteq  Q \times (\Sigma_o \times [0,N_o])^{\le N_o+1}\times [0,N_{l,o}] \times \Pi \times (\Pi \times [0,N_c])^{\le N+1}  \times [0,N_{l,c}] \times X$ is the state space; $\tilde{\Sigma}\subseteq \Sigma \cup \Sigma^o \cup\Sigma^c \cup \Sigma^f \cup \Sigma^g$ is the event set; $\tilde{q}_{0}=(q_0, \varepsilon, 0, S(\varepsilon), \varepsilon, 0, x_0)$ is the initial state; and the transition function $\tilde{\delta}: \tilde{Q} \times \tilde{\Sigma} \rightarrow \tilde{Q}$ is defined as follows:

\begin{itemize} 

  \item For all $\tilde{q}=(q,x,n,\phi,y,m,p)\in \tilde{Q}$ and all $\sigma \in \Sigma$,
\Lfteqn\label{Eq1}
  \tilde{\delta}(\tilde{q}, \sigma)= \begin{cases}
    \tilde{q}' & \text{if}\ \delta(q,\sigma)! \wedge \sigma \in \phi \\
&\wedge  \textbf{NUM}(x^+) \le N_o\\
& \wedge \textbf{NUM}(y^+) \le N_c\\
    \mbox{undefined} & \mbox{otherwise,}
  \end{cases}
  \Ndeqn  
with $\tilde{q}' =(q',x',n',\phi',y',m',p')$, where (i) $q'=\delta(q,\sigma)$, (ii) if $\sigma \in \Sigma_o$, $x'=x^+(\sigma,0)$, and if $\sigma \in \Sigma_{uo}$, $x'=x^+$, (iii) $n'=n$, (iv) $\phi'=\phi$, (v) $y'=y^+$, (vi) $m'=m$, and (vii)  $p'=p$.

  \item For all $\tilde{q}=(q,x,n,\phi,y,m,p)\in \tilde{Q}$ and all $o(i) \in \Sigma^o$,  if $x=\varepsilon$, $\tilde{\delta}(\tilde{q},o(i))$ is not defined, and if $x \neq \varepsilon$, we write $x=(\sigma_1,,n_1) \cdots (\sigma_k,n_k)$ for $\sigma_j \in \Sigma_o$ and $n_j \in [0,N_o]$, and then,
\Lfteqn\label{Eq2}
  \tilde{\delta} (\tilde{q}, o(i))= \begin{cases}
    \tilde{q}' & \text{if}\ i \in [1,k] \wedge n+1 \le N_{l,o}\\
    \mbox{undefined} & \mbox{otherwise,}
  \end{cases}
  \Ndeqn  
with $\tilde{q}' =(q',x',n',\phi',y',m',p')$, where  (i) $q'=q$, (ii) $$x'=(\sigma_1,n_1) \cdots (\sigma_{i-1},n_{i-1})(\sigma_{i+1},n_{i+1})\cdots (\sigma_k,n_k),$$
(iii) $n'=n+1$, (iv) $\phi'=\phi$, (v) $y'=y$,  (vi) $m'=m$, and (vii) $p'=p$.

\item For all $\tilde{q}=(q,x,n,\phi,y,m,p)\in \tilde{Q}$ and all $f(\sigma) \in \Sigma^f$, if $x=\varepsilon$, then $\tilde{\delta}(\tilde{q},f(\sigma))$ is not defined, 
and if $x \neq \varepsilon$, we write $x=(\sigma_1,n_1) \cdots (\sigma_k,n_k)$ for $\sigma_j \in \Sigma_o$ and $n_j \le N_o$, and then,
\Lfteqn\label{Eq3}
  \tilde{\delta} (\tilde{q}, f(\sigma))= \begin{cases}
    \tilde{q}' & \text{if}\ \sigma=\sigma_1\\
    \mbox{undefined} & \mbox{otherwise,}
  \end{cases}
  \Ndeqn  
with $\tilde{q}' =(q',x',n',\phi',y',m',p')$, where  (i) $q'=q$, (ii) $x'=(\sigma_2,n_2) \cdots (\sigma_k,n_k)$, (iii) $n'=0$, (iv) $\phi'=\phi$, (v) $y'=y(\gamma(\xi(p,\sigma)),0)$, (vi) $m'=m$, and (vii) $p'=\xi(p,\sigma)$.  

\item For all $\tilde{q}=(q,x,n,\phi,y,m,p)\in \tilde{Q}$ and all $c(i) \in \Sigma^c$, if $y=\varepsilon$, $\tilde{\delta}(\tilde{q},c(i))$ is not defined, 
and if $y \neq \varepsilon$, we write $y=(\pi_1,m_1) \cdots (\pi_h,m_h)$ for $\pi_j \in \Pi$ and $m_j \le N_c$, and then,
\Lfteqn\label{Eq4}
  \tilde{\delta}(\tilde{q}, c(i))= \begin{cases}
    \tilde{q}' & \text{if}\ i \in [1,h] \wedge m+1 \le N_{l,c}\\
    \mbox{undefined} & \mbox{otherwise,}
  \end{cases}
  \Ndeqn  
with $\tilde{q}' =(q',x',n',\phi',y',m',p')$, where  (i) $q'=q$, (ii) $x'=x$, (iii) $n'=n$, (iv) $\phi'=\phi$,  (v) $$y'=(\pi_1,m_1) \cdots (\pi_{i-1},m_{i-1})(\pi_{i+1},m_{i+1})\cdots (\pi_h,m_h),$$  (vi) $m'=m+1$, and (vii) $p'=p$.

  \item For all $\tilde{q}=(q,x,n,\phi,y,m,p)\in \tilde{Q}$ and all $g(\pi) \in \Sigma^g$, if $y=\varepsilon$, $\tilde{\delta}(\tilde{q},g(\pi))$ is not defined, 
and if $y \neq \varepsilon$, we write $y=(\pi_1,m_1) \cdots (\pi_h,m_h)$ for $\pi_j \in \Pi$ and $m_j \le N_c$, and then,
\Lfteqn\label{Eq55}
  \tilde{\delta} (\tilde{q}, g(\pi))= \begin{cases}
    \tilde{q}' & \text{if}\ \pi=\pi_1\\
    \mbox{undefined} & \mbox{otherwise,}
  \end{cases}
  \Ndeqn  
with $\tilde{q}' =(q',x',n',\phi',y',m',p')$, where  (i) $q'=q$, (ii) $x'=x$, (iii) $n'=n$, (iv) $\phi'=\pi$, (v) $y'=(\pi_2,m_2) \cdots (\pi_h,m_h)$, (vi) $m'=0$, and (vii) $p'=p$.

\end{itemize}

We interpret the construction of $G_S$ as follows. 
In (\ref{Eq1}), $\sigma$ can occur at $\tilde{q}=(q,x,n,\phi,y,m,p)\in \tilde{Q}$ iff $\sigma$ is defined at $q$ in $G$, i.e., $\delta(q,\sigma)!$, the occurrence of $\sigma$ is allowed by the control action in use, i.e., $\sigma \in \phi$, and the observation delays and control delays after the occurrence of $\sigma$ are no larger  than $N_o$ and $N_c$, respectively, i.e.,  {$\textbf{NUM}(x^+)\le N_o  \wedge \textbf{NUM}(y^+) \le N_c$}.
If $\sigma$ occurs at $\tilde{q}$,   since $q$ is used to track the state that the plant is in, we have $q'=\delta(q,\sigma)$. 
Furthermore, if an unobservable $\sigma \in \Sigma_{uo}$ occurs at $\tilde{q}$,  the sequence delayed at the observation channel still is $x$ but all the numbers in $x$ (if $x \neq \varepsilon$) should add 1 for counting the observation delays.
Therefore, we set $x=x^+$ if $\sigma \in \Sigma_{uo}$ in (\ref{Eq1}).
However, if $\sigma \in \Sigma_o$, by FIFO, $(\sigma,0)$ should be added to the end of $x$ for tracking the new observable event occurrence, which is illustrated by $x'=x^+(\sigma,0)$ if $\sigma \in \Sigma_{o}$ in (\ref{Eq1}).
Meanwhile, after the occurrence of $\sigma$, the numbers in $y$ should add 1 for recording the control delays. Hence, we set $y'=y^+$ in (\ref{Eq1}).
For the remaining components in $\tilde{q}$,  the state of the supervisor can be updated only when a new event is communicated, and the control action in use can be updated only when a new control action is executed.
Therefore, $p$ and $\phi$ have no change after the occurrence of $\sigma$, i.e., $p'=p$ and $\phi'=\phi$.
Since there are no observation losses and control losses,  $n'=n$ and $m'=m$.

For any $\tilde{q}=(q,x,n,\phi,y,m,p)\in \tilde{Q}$, if $x=\varepsilon$, the observation channel is empty.
Thus, no observable events can be lost or communicated.
As we can see, if $x=\varepsilon$, $o(i)\in \Sigma^o$ is not defined at $\tilde{q}$ in (\ref{Eq2}), and $f(\sigma) \in \Sigma^f$ is not defined at $\tilde{q}$ in (\ref{Eq3}).
Otherwise, if $x=(\sigma_1,n_1)\cdots(\sigma_k,n_k) \neq \varepsilon$,  by assumption, all the observable events queued at the observation channel may be lost if the consecutive observation losses are not larger than $N_{l,o}$ after the observation loss.
Therefore, if $n+1 \le N_{l,o}$,  $o(i)$ is defined at $\tilde{q}$ for all $i\in [1,k]$  in (\ref{Eq2}).
When $o(i)$ occurs at $\tilde{q}$, we remove $(\sigma_i,n_i)$ from $x$ and update $n$ to $n+1$ to record the observation loss.
On the other hand, since $x=(\sigma_1,n_1)\cdots(\sigma_k,n_k)$, by FIFO,  $f(\sigma)$ is defined at $\tilde{q}$ iff  $\sigma=\sigma_1$  in (\ref{Eq3}).
When the occurrence of $\sigma_1$ is communicated, the remaining sequence to be communicated is $\sigma_2 \cdots \sigma_k$.
Therefore, we have $x'=(\sigma_2,n_2)\cdots(\sigma_k,n_k)$ in (\ref{Eq3}).
Following the communication of $\sigma_1$, the state of the supervisor is updated to $\xi(p,\sigma)$, and the supervisor
sends $\gamma(\xi(p,\sigma))$ to the actuator of the plant, which are illustrated, respectively by $p'=\xi(p,\sigma)$ and $y'=y(\gamma(\xi(p,\sigma)),0)$ in (\ref{Eq3}).
Furthermore, since $n$ is used to count the number of consecutive observation losses, we reset $n$ to $0$ after a new event communication, as $n'=0$ in (\ref{Eq3}).



For any $\tilde{q}=(q,x,n,\phi,y,m,p)\in \tilde{Q}$, if $y=\varepsilon$, the control channel is empty, and no control actions can be lost or  executed. 
Thus, if $y=\varepsilon$, then $c(i) \in \Sigma^c$ is not defined at $\tilde{q}$ in (\ref{Eq4}), and $g(\pi) \in \Sigma^g$  is not defined at $\tilde{q}$ in (\ref{Eq55}).
Otherwise, if $y=(\pi_1,m_1) \cdots (\pi_h,m_h) \neq \varepsilon$, by assumption, all the control actions queued at the control channel may be lost if the consecutive control losses are no larger than $N_{l,c}$ after the control loss.
Thus, if $m+1\le N_{l,c}$,  $c(i)$ is defined at $\tilde{q}$ for all $i\in [1,h]$ in (\ref{Eq4}).
When $c(i)$ occurs at $\tilde{q}$, we remove $(\pi_i,m_i)$ from $y$ and update $m$ to $m+1$  to record the control loss.
On the other hand, since  $y=(\pi_1,m_1) \cdots (\pi_h,m_h) \neq \varepsilon$, by FIFO, $g(\pi)$ is defined at $\tilde{q}$ iff $\pi=\pi_1$.
When $g(\pi)$ occurs at $\tilde{q}$, the control action that is taking effect becomes $\pi$, and the control actions queued at the control channel are $\pi_2\cdots\pi_h$, which are illustrated, respectively by  $\phi'=\pi$ and $y'=(\pi_2,m_2) \cdots (\pi_h,m_h)$ in (\ref{Eq55}). 
Since $m$ is used to count the number of consecutive control losses, we reset $m$ to $0$ in (\ref{Eq55}) when a new control action is executed.
Example \ref{Ex2} further illustrates how to construct  $G_S$.




\begin{figure}
	\begin{center}
		\includegraphics[width=8cm]{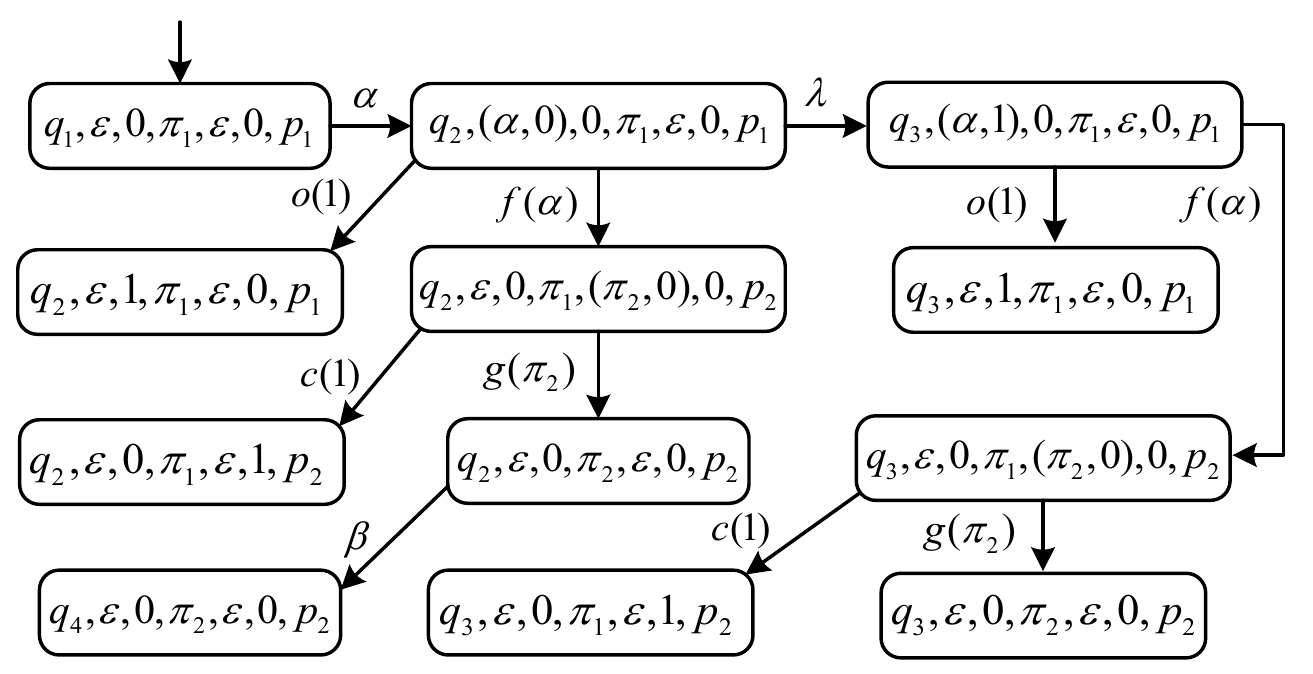}    
		\caption{Communication automaton $G_S$ in Example \ref{Ex2}} 
		\label{Fig3}
	\end{center}
\end{figure}

\begin{example}\label{Ex2}
Again, we consider the $G$ and the $S$ depicted in Fig. \ref{Fig21} and  Fig. \ref{Fig22}, respectively.
As shown in Example \ref{Ex1}, $\Sigma_c=\Sigma$ and $\Sigma_o=\{\alpha\}$.
Let $N_c=N_o=N_{l,c}=N_{l,o}=1$.
The communication automaton $G_S$ is constructed in Fig. \ref{Fig3}.

Let us consider the initial state $\tilde{q}_0=(q_1,\varepsilon, 0,\pi_1,\varepsilon,0,p_1)$.
Since the control channel is empty, no control actions can be executed or lost.
Hence, all the events in $\Sigma^c$ and $\Sigma^g$ are not defined at $\tilde{q}_0$.
Similarly, all the events in $\Sigma^o$ and $\Sigma^f$ are also not defined at $\tilde{q}_0$ since the observation channel is empty.
By Fig.\ref{Fig21}, $\delta(q_1,\alpha)=q_2$.
Moreover, since $\alpha \in \Sigma_o$, $\alpha \in \pi_1$, and $\textbf{NUM}(\varepsilon^+)=0 \le N_c, N_o$, by (\ref{Eq1}),  we have $\tilde{\delta}(\tilde{q}_0,\alpha)=\tilde{q}_1=(q_2,(\alpha,0),0,\pi_1,\varepsilon,0,p_1)$.

Next, we consider state $\tilde{q}_1=(q_2,(\alpha,0),0,\pi_1,\varepsilon,0,p_1)$.
By Fig. \ref{Fig21}, $\delta(q_2,\beta)=q_4$ and  $\delta(q_2,\lambda)=q_3$. 
Since $\beta \notin \pi_1$, by (\ref{Eq1}), $\beta$ cannot occur at $\tilde{q}_1$.
However, since $\lambda \in \Sigma_{uo}$, $\lambda \in \pi_1$, $\textbf{NUM}((\sigma,0)^+)=1 \le N_o$, and $\textbf{NUM}(\varepsilon^+)=0 \le N_c$, by (\ref{Eq1}), $\tilde{\delta}(\tilde{q}_1,\lambda)=\tilde{q}_1=(q_3,(\alpha,1),0,\pi_1,\varepsilon,0,p_1)$.
Moreover, since $(\alpha,0) \neq \varepsilon$, by (\ref{Eq3}), $f(\alpha)$ is defined at $\tilde{q}_1$.
When $\sigma$ is communicated, $S$ moves to state $\xi(p_1,\alpha)=p_2$, and a new control action $\gamma(p_2)=\pi_2=\{\beta\}$ is issued.
By (\ref{Eq3}),  $\tilde{\delta}(\tilde{q}_1,f(\alpha))=(q_2,\varepsilon,0,\pi_1,(\pi_2,0),0,p_2)$.
In addition, since $(\alpha,0) \neq \varepsilon$ and $0+1=1\le N_{l,o}$, the event occurrence of $\alpha$ may be lost from the observation channel at $\tilde{q}_1$.
By (\ref{Eq2}), $\tilde{\delta}(\tilde{q}_1,o(1))=(q_2,\varepsilon,1,\pi_1,\varepsilon, 0, p_1)$.
In this way, we can construct $G_S$.

\end{example}

\begin{figure}
	\begin{center}
		\includegraphics[width=8cm]{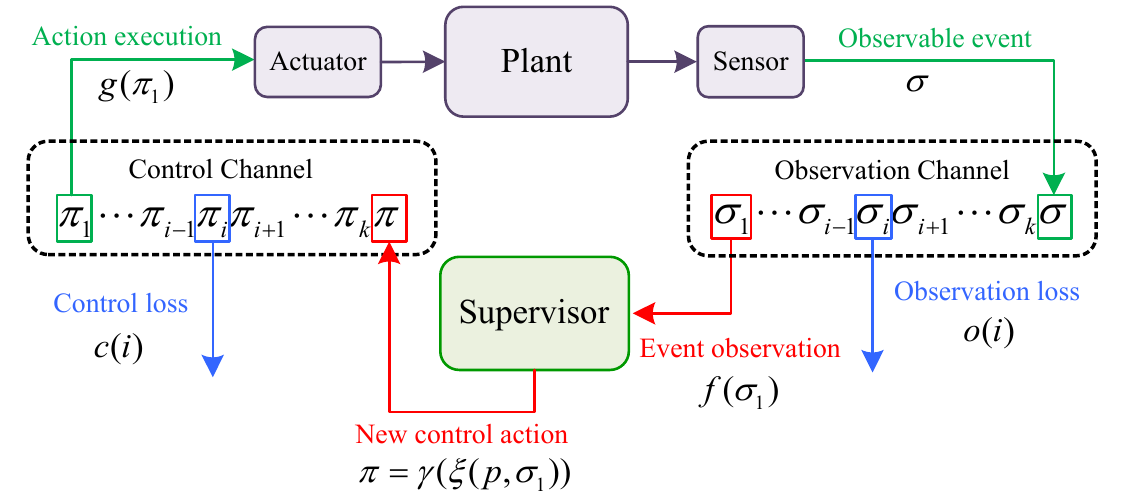}    
		\caption{The interaction process between the plant and the supervisor.} 
		\label{process}
	\end{center}
\end{figure}

\begin{remark}
\textcolor{blue}{The interaction process between the plant and the supervisor is illustrated in Fig. \ref{process}. 
When a new observable event $\sigma$ occurs in the plant, it is immediately pushed into the observation channel.
When the first event $\sigma_1$ queued at the observation channel is delivered to the supervisor, a new control action $\pi=\gamma(\xi(p,\sigma_1))$ can be immediately issued and inserted into the control channel.
The first control action $\pi_1$ queued at the control channel cannot be executed  until it is delivered to the actuator of the plant.
Both the control actions delayed at the control channel and the observable events delayed at the observation channel may be lost.
The supervisor has no idea what observable events are now queued at the observation channel and what control actions are now queued at the control channel. 
As will be shown in the next section, the supervisor makes state estimation based on only the observable events that have been communicated to it.}
\end{remark}

Given a string $\mu \in \tilde{\Sigma}^*$, let $\psi(\mu)$ and $\psi^f(\mu)$ be the string obtained by removing all the events in $\tilde{\Sigma} \setminus \Sigma$ and $\tilde{\Sigma} \setminus \Sigma_f$  from $\mu$, respectively, without changing the order of the remaining events.
Define $f^{-1}$ as, for all $f(\sigma) \in \Sigma_f$, $f^{-1}(f(\sigma))=\sigma$.
We extend $\psi$, $\psi^f$, and $f^{-1}$ to a set of strings in the usual way.
We consider $\mu=\alpha f(\alpha) c(1) \in \mathcal{L}(G_S)$ in Fig. \ref{Fig3}.
By definitions, $\psi(\mu)=\alpha \in \mathcal{L}(G)$, $\psi^f(\mu)=f(\alpha)$, and $f^{-1}(\psi^f(\mu))=\alpha$.

Intuitively, 1) $\psi(\mathcal{L}(G_S))$  specifies all the languages that can be generated under  $S$, and 2) $f^{-1}(\psi^f(\mathcal{L}(G_S)))$ specifies all the behaviors that can be observed by the networked supervisor.
We formally prove them in the following proposition.

\begin{proposition}\label{Prop1}
Given a $\mu \in \mathcal{L}(G_S)$, let us write $\tilde{\delta}(\tilde{q}_0,\mu)=(q,x,n,\phi,y,m,p)$.
Then, we have (i) $q=\delta(q_0,\psi(\mu))$ and (ii) $p=\xi(x_0,f^{-1}(\psi^f(\mu)))$.
\end{proposition}

By Proposition \ref{Prop1}, the ``dynamics'' of the controlled system can be simply obtained by removing all the events in $\tilde{\Sigma}\setminus \Sigma$ from the sequences generated by $G_S$.

\begin{definition}\label{Def1}
Given a system $G$ and a networked supervisor $S=(A,\gamma)$ with $A=(X,\Sigma_o,\xi,x_0)$, we construct $G_S$ as described above.
The language that may be generated by the controlled system under the communication delays and losses, denoted by $\mathcal{L}(S/G)$, is defined as $\mathcal{L}(S/G)=\psi(\mathcal{L}(G_S))$.
\end{definition}

\begin{proposition}\label{Prop2}
Given two networked supervisors $S_i=(A_i,\gamma_i)$ with $A_i=(X_i,\Sigma_o,\xi_i,x_{0,i})$, $i \in \{1,2\}$,  we construct $G_{S_i}$ as described above.
Then, if $S_1 \subseteq S_2$,  $\mathcal{L}(S_1/G) \subseteq \mathcal{L}(S_2/G)$.
\end{proposition}

Proposition \ref{Prop2} shows that the more events a supervisor enables, the larger language the controlled system generates.
We now introduce the definition of the networked state estimate (NSE) of the controlled system for a communicated string.
\begin{definition}\label{Def2}
Given a system $G$ and a networked supervisor $S$ defined over $\Sigma_o^*$, for any $t \in f^{-1}(\psi^f(\mathcal{L}(G_S)))$, define
\begin{align}\label{Eq6}
\mathcal{E}_S(t)=&\{q\in Q:(\exists \mu \in \mathcal{L}(G_S))\notag\\
&q=\delta(q_0,\psi(\mu))\wedge t=f^{-1}(\psi^f(\mu))\}, 
\end{align}
as the NSE of $t$ under $S$, which is the set of all the possible states that the plant $G$ may be in after observing $t$ (subject to communication delays and losses) under $S$.
\end{definition}

If $S$ is given beforehand, we can calculate $\mathcal{E}_S(t)$ by constructing an observer of $G_S$ with the set of observable events $\Sigma^f$ \cite{lafortune07book}.
However, when we solve the {supervisor synthesis problem}, $S$ is unknowable. 
The state estimate should be calculated online immediately after each new observation without using the future observations and controls  \cite{yin16tac1,yin16tac2,yin17tac,yin18tac}.
This is exactly the problem we want to solve in this paper.

\section{Online networked state estimation}

{In this section, we discuss how to produce online estimates of the states of a controlled system under communication delays and losses.
\textcolor{blue}{To determine which state the controlled system is in, we should estimate  not only the states of the plant, but also the observable event occurrences delayed at the observation channel  as well as  the control actions delayed at the control channel.
This is because all of them can affect the behaviors of the controlled system.}
To this end, we introduce the notions of \textsl{observation channel configuration} and \textsl{control channel configuration} as follows.}

\begin{definition}\label{DefI}
The {observation channel configuration} is defined as: $\theta_o=((\sigma_1,n_1)\cdots (\sigma_k,n_k), n),$
where $(\sigma_1,n_1)\cdots (\sigma_k,n_k)\in (\Sigma_o \times [0,N_o])^{\le N_o+1}$ is sequence of pairs such that $\sigma_1\cdots\sigma_k$ is a sequence of observable events currently delayed at the observation channel (in the same order as they were generated) and $n_i$ is the number of event occurrences since $\sigma_i$ has occurred, and $n\in [0,N_{l,o}]$ tracks the number of consecutive observation losses.
\end{definition}
We denote by $\Theta_o \subseteq (\Sigma_o \times [0,N_o])^{\le N_o+1} \times [0,N_{l,o}]$ the set of all the possible observation channel configurations.
By Definition \ref{DefI}, $\theta_o$ can be updated if one of the following three behaviors happens: an event occurs, an observable event is communicated, or an observable event is lost.
To update $\theta_o$, we  define three operations as follows.
Given a $\theta_o=(x,n)\in \Theta_o$,
\begin{enumerate} 
\item  if an event $\sigma  \in \Sigma$ occurs,  $x$ should be updated to $x^+$ immediately to count the observation delays.
Meanwhile, if $\sigma \in \Sigma_o$, by FIFO, we still need to add $(\sigma,0)$ to the end of $x^+$ to record the new event occurrence. 
Formally, for any $\theta_o=(x,n)\in \Theta_o$ and any  $\sigma \in \Sigma$, we define $\textbf{IN}^{obs}(\theta_o,\sigma)=(x',n')$, where  if $\sigma \in \Sigma_o$, $x'=x^+(\sigma,0)$ and $n'=n$, and  if $\sigma \in \Sigma_{uo}$, $x'=x^+$ and $n'=n$;

\item if a new  $\sigma \in \Sigma_o$ is communicated, by FIFO, $\sigma$ is the first event queued at the observation channel.
If we write $x=(\sigma_1,n_1)\cdots (\sigma_k,n_k)$, then $\sigma=\sigma_1$ and the remaining events delayed at the observation channel are $\sigma_2\cdots\sigma_k$.
Additionally, since $n$ is used to track the number of consecutive  observation losses, we reset $n$ to $0$ after a new event communication. 
Therefore, for any $\theta_o=(x,n)\in \Theta_o$, if $x=(\sigma_1,n_1)\cdots (\sigma_k,n_k)\neq \varepsilon$, we define $\textbf{OUT}^{obs}(\theta_o)=(x',n')$, where $x'=(\sigma_2,n_2)\cdots (\sigma_k,n_k)$ and $n'=0$;

\item if the $i$th event in the observation channel is lost, we should remove it from $x$.
Meanwhile, since a new observation loss occurs, the number of consecutive observation losses should be updated to $n+1$.
Thus, for any $\theta_o=(x,n) \in \Theta_o$ with $x=(\sigma_1,n_1)\cdots (\sigma_k,n_k) \neq \varepsilon$ and any $i \in [1,k]$, we define $\textbf{LOSS}^{obs}(\theta_o,i)=(x',n'),$ where $x'=(\sigma_1,n_1)\cdots (\sigma_{i-1},n_{i-1}) (\sigma_{i+1},n_{i+1})\cdots(\sigma_k,n_k)$ and $n'=n+1$.
\end{enumerate} 

\begin{definition}\label{DefII}
The control channel configuration is defined as: $\theta_c=(\phi, y=(\pi_1,m_1)\cdots (\pi_h,m_h), m),$
where $\phi\in \Pi$ is the control action in use, $(\pi_1,m_1)\cdots (\pi_h,m_h)\in (\Pi \times [0,N_c])^{\le N+1}$ is a sequence of pairs such that $\pi_1\cdots\pi_h$ are control actions currently queued at the control channel, and $m_i$ is the number of event occurrences since control action $\pi_i$ has been issued, and $m \in [0,N_{l,c}]$ counts the number of consecutive control losses.
\end{definition}
We denote by $\Theta_c \subseteq \Pi \times (\Pi \times [0,N_c])^{\le N+1} \times [0,N_{l,c}]$ the set of all the possible control channel configurations.
By Definition \ref{DefII}, $\theta_c$ can be changed if one of the following four behaviors happens: (i) a new control action is issued, (ii) a new control action is executed, (iii) a control action is lost, and (iv) a new event occurs.
To update $\theta_c$, we next define four operations as follows.
Given a  $\theta_c=(\phi,y,m)\in \Theta_c$,
\begin{enumerate}

\item if a new control action $\pi \in \Pi$ is issued, by FIFO, we need to add $(\pi,0)$ to the end of $y$.
Formally, for any $\theta_c=(\phi,y,n)\in \Theta_c$ and any $\pi \in \Pi$, we define $\textbf{IN}^{ctr}(\theta_c,\pi)=(\phi',y',m'),$ where $\phi'=\phi$, $y'=y(\pi,0)$, and $m'=m$;

\item if a new control action $\pi \in \Pi$ is executed, by FIFO,  $\pi$ is the first control action queued at the control channel.
After execution, the control action that is taking effect would be $\pi$.
Meanwhile, since $m$ is used to track the number of consecutive control losses, we need to reset $m$ to $0$ after a new control action execution. 
Formally, for any $\theta_c=(\phi,y,m)\in \Theta_c$, if $y=(\pi_1,m_1)\cdots (\pi_h,m_h)\neq \varepsilon$, define $\textbf{OUT}^{ctr}(\theta_c)=(\phi',y',m')$, where $\phi'=\pi_1$ and $y'=(\pi_2,m_2)\cdots (\pi_h,m_h)$ and $m'=0$;

\item if the $i$th control action in the control channel is lost, by definition, we need to remove it from $y$.
Meanwhile, since a new control loss occurs, the number of consecutive control losses becomes $m+1$.
Formally, for any $\theta_c=(\phi,y,m)\in \Theta_c$ with $y=(\pi_1,m_1)\cdots (\pi_h,m_h)\neq \varepsilon$ and any $i \in [1,h]$, define $\textbf{LOSS}^{ctr}(\theta_c,i)=(\phi', y',m'),$  where $y'=(\pi_1,m_1)\cdots (\pi_{i-1},m_{i-1})(\pi_{i+1},m_{i+1})\cdots (\pi_h,m_h)$ and $\phi'=\phi$ and $m'=m+1$;

\item if a new event occurs in $G$, all the natural numbers in $y$ (if $y \neq \varepsilon$) should increment for tracking the control delays.
Hence, for any $\theta_c=(\phi,y,m)\in \Theta_c$, define $\textbf{PLUS}(\theta_c)=(\phi', y',m')$, where $\phi'=\phi$ and $y'=y^+$ and $m'=m$.
\end{enumerate} 

Given a $\theta_o=(x,n)\in \Theta_o$, let $[\theta_o]_1=x$ and $[\theta_o]_2=n$ be the first and second components of $\theta_o$, respectively.
Similarly, given a $\theta_c=(\phi,y,m)\in \Theta_c$, let $[\theta_c]_1=\phi$, $[\theta_c]_2=y$, and $[\theta_c]_3=m$  be the first,  second, and third components of $\theta_c$, respectively.

As mentioned above, in addition to  $q\in Q$, we also need to estimate $\theta_o\in \Theta_o$ and $\theta_c\in \Theta_c$ since they can affect the future behaviors of the controlled system.
Thus, we denote each state of the controlled system by a triplet $(q,\theta_o, \theta_c) \in Q  \times \Theta_o \times \Theta_c$.
We call such a state an augmented state.
Next, we  show how to update the augmented state estimate upon each new communication.
The procedure can be briefly summarized as repeatedly executing the following two steps.

\emph{Step 1:}
Let $Z \subseteq Q \times \Theta_o \times \Theta_c$ be a set of augmented states calculated immediately after a new observation or the initial $Z=\emptyset$\footnote{\textcolor{blue}{By assumption, the plant does not work until it is initialized. Thus, before the initial control action is executed (the plant starts to work), we let $Z=\emptyset$.}}. 
The delayed unobservable reach of $Z$ under an admissible control action $\pi\in \Pi$, denoted by $\text{DUR}(Z,\pi)$, is defined as follows.

\begin{enumerate} 
  \item Initially, if $Z=\emptyset$, we have  
 \begin{align}\label{Eq7-1}
  (q_0, (\varepsilon,0), (\pi,\varepsilon,0)) \in  \text{DUR}(Z,\pi).
 \end{align}
Otherwise, if $Z \neq \emptyset$,  for all $(q,\theta_o,\theta_c) \in Z$,  
  \begin{align}\label{Eq7}
  (q, \theta_o, \textbf{IN}^{ctr}(\theta_c,\pi)) \in  \text{DUR}(Z,\pi);
 \end{align}
 \item Then, we repeatedly apply the following operations until convergence is achieved.
\begin{itemize}
  \item For all $(q,\theta_o,\theta_c) \in \text{DUR}(Z,\pi)$, if $\delta(q,\sigma)!$ and $\sigma \in [\theta_c]_1$ and $\textbf{NUM}([\theta_o]_1^+)\le N_o$ and $\textbf{NUM}([\theta_c]_2^+)\le N_c$, 

\begin{align}\label{Eq8}
     (\delta(q,\sigma), \textbf{IN}^{obs}(\theta_o,\sigma),  \textbf{PLUS}(\theta_c)) \in  {\text{DUR}}(Z,\pi);
 \end{align}

\item For all $(q,\theta_o,\theta_c) \in \text{DUR}(Z,\pi)$, if $[\theta_c]_2 \neq \varepsilon$, then 
 \begin{align}\label{Eq9}
  (q, \theta_o, \textbf{OUT}^{ctr}(\theta_c)) \in  \text{DUR}(Z,\pi);
\end{align}
 \item For all $(q,\theta_o,\theta_c) \in \text{DUR}(Z,\pi)$, if  $[\theta_o]_1 \neq \varepsilon$ and $[\theta_o]_2+1 \le N_{l,o}$, then for all $i \in [1, |[\theta_o]_1|]$
 \begin{align}\label{Eq10}
 (q,\textbf{LOSS}^{obs}(\theta_o,i), \theta_c) \in  \text{DUR}(Z,\pi);
\end{align}
\item For all $(q,\theta_o,\theta_c) \in \text{DUR}(Z,\pi)$, if $[\theta_c]_2 \neq \varepsilon$ and $[\theta_c]_3+1 \le N_{l,c}$, then for all $i \in [1, |[\theta_c]_2|]$
 \begin{align}\label{Eq11}
(q,\theta_o, \textbf{LOSS}^{ctr}(\theta_c,i)) \in  \text{DUR}(Z,\pi).
\end{align}
\end{itemize}
\end{enumerate} 

\begin{remark}
In the context of networked DESs, only when ``an observable event is communicated'' is it observable. 
The behaviors of ``an event (observable or not) occurs'', ``a control action is executed'', ``an observable event is lost'' and ``a control action is lost'' are all unobservable.
They are considered by (\ref{Eq8}), (\ref{Eq9}), (\ref{Eq10}), and (\ref{Eq11}), respectively.
Operation (\ref{Eq7-1}) is used to set the initial control action to be $\pi$, and operation (\ref{Eq7}) is used to insert the newly issued $\pi$ into the control channel.
\end{remark}

Intuitively, ${\text{DUR}}(Z,\pi)$ consists of all the augmented states that can be reached from augmented states  in $Z$ in an ``unobservable'' way.
\textcolor{blue}{Specifically, if $Z=\emptyset$, the plant has not been initialized, and $\pi$ is the initial control action. By assumption, $\pi$ can be executed without any delays and losses since it has been deployed in the actuator of the plant beforehand. Thus, we have $(q_0, (\varepsilon,0), (\pi,\varepsilon,0)) \in  \text{DUR}(Z,\pi)$ in (\ref{Eq7-1}). Otherwise, if $Z \neq \emptyset$, by FIFO, we should add $\pi$ to the end of the control channel. Thus, for all $(q,\theta_o,\theta_c) \in Z$,  we have $ (q, \theta_o, \textbf{IN}^{ctr}(\theta_c,\pi)) \in  \text{DUR}(Z,\pi)$ in (\ref{Eq7}).}
For any $(q,\theta_o,\theta_c) \in \text{DUR}(Z,\pi)$,   an event $\sigma$ can occur at $q$ iff (i) $\sigma$ is active at $q$, i.e., $\delta(q,\sigma)!$, (ii) $\sigma$ is allowed to occur by the control action that is taking effect, i.e., $\sigma \in [\theta_c]_1$, (iii) after the occurrence of $\sigma$, the control delays and the observation delays are no larger than $N_c$ and $N_o$, i.e., $\textbf{NUM}([\theta_c]_2^+) \le N_c  \wedge \textbf{NUM}([\theta_o]_1^+) \le N_o$. 
If  $\sigma$ occurs at $q$, then (i) the plant moves to state $\delta(q,\sigma)$, (ii) $\theta_o$ is updated to $\textbf{IN}^{obs}(\theta_o,\sigma)$ to record the occurrence of $\sigma$, and (iii) $\theta_c$ is updated to $\textbf{PLUS}(\theta_c)$  to track the control delays. 
This is illustrated by  $ (\delta(q,\sigma), \textbf{IN}^{obs}(\theta_o,\sigma),  \textbf{PLUS}(\theta_c)) \in  \text{DUR}(Z,\pi)$ in (\ref{Eq8}).
Furthermore, for any $(q,\theta_o,\theta_c) \in \text{DUR}(Z,\pi)$, a control action can be executed if the queue of control actions delayed  at the control channel is not empty, i.e., $[\theta_c]_2 \neq \varepsilon$.
We write $[\theta_c]_2=(\pi_1,m_1) \cdots (\pi_h,m_h)$ for $\pi_i \in \Pi$ and $m_i \in [0,N_c]$.
When a new control action $\pi$ is executed, by FIFO,  the control actions delayed at the control channel are $\pi_2\cdots\pi_h$, and the control action that is taking effect becomes $\pi=\pi_1$.
This is illustrated by $ (q, \theta_o, \textbf{OUT}^{ctr}(\theta_c)) \in  \text{DUR}(Z,\pi)$ in (\ref{Eq9}).
For any $(q,\theta_o,\theta_c) \in \text{DUR}(Z,\pi)$, an observable event occurrence can be lost if the observation channel is not empty, and the consecutive observation losses are no larger than $N_{l,o}$ after the observation loss, i.e., $[\theta_o]_1 \neq \varepsilon$ and $[\theta_o]_2 +1 \le N_{l,o}$. 
When the $i$th, $i \in [1,|[\theta_o]_1|]$ observable event is lost from the observation channel, we have $(q,\textbf{LOSS}^{obs}(\theta_o,i), \theta_c) \in  \text{DUR}(Z,\pi)$ in (\ref{Eq10}).
Similarly, if $[\theta_c]_2 \neq \varepsilon$ and $[\theta_c]_3+1 \le N_{l,c}$, we know the $i$th, $i \in [1,|[\theta_c]_2|]$  control action may be lost from the control channel.
Thus, $(q,\theta_o, \textbf{LOSS}^{ctr}(\theta_c,i)) \in  \text{DUR}(Z,\pi)$ in (\ref{Eq11}).



\emph{Step 2:} Let $Z \subseteq Q \times \Theta_o \times \Theta_c$ be a given augmented state estimate. 
The delayed observable reach of  $Z$ under an observable event $\sigma \in \Sigma_o$, denoted by $\text{DOR}(Z,\sigma)$,  is defined as:
\begin{align}\label{Eq12}
\text{DOR}(Z,\sigma)=&\{(q,\textbf{OUT}^{obs}(\theta_o),\theta_c):(\exists (q,\theta_o,\theta_c) \in Z)\notag\\
&[\theta_o]_1=(\sigma_1,n_1)\cdots(\sigma_k,n_k) \neq \varepsilon \wedge \sigma_1=\sigma\}.
\end{align}

 $\text{DOR}(Z,\sigma)$ includes all the augmented states that can be reached from the augmented states in  $Z$ following a new communication of $\sigma$. 
By FIFO, an observable event can be communicated iff it is the first event queued at the observation channel.
Hence, we only consider all the $\sigma \in \Sigma_o$ such that there exists $(q,\theta_o,\theta_c) \in Z$ with $[\theta_o]_1=(\sigma_1,n_1)\cdots(\sigma_k,n_k) \neq \varepsilon \wedge \sigma_1=\sigma$. 
When $\sigma$ is communicated, we remove $(\sigma_1,n_1)$ from $\theta_o$.
Thus, after communication of $\sigma$, $\theta_o$ is updated to $\textbf{OUT}^{obs}(\theta_o)$.
We assume that $\text{DOR}(Z,\sigma)$ is updated immediately after a new observation $\sigma$ but before the next control action is issued.
Therefore, we keep $\theta_c$ unchanged in (\ref{Eq12}).

For a communicated  string $t\in \Sigma_o^*$, let the set of augmented states calculated  by alternatively applying \emph{Step 1} and \emph{Step 2} be the augmented state estimate for $t$.
Formally, 
\begin{definition}\label{Def5}
Given a system $G$ and a networked supervisor $S$ defined over $\Sigma_o^*$, for any $t \in f^{-1}(\psi^f(\mathcal{L}(G_S)))$, let $\tilde{\mathcal{E}}_S(t)$ be the augmented state estimate calculated by alternatively applying $\textup{DUR}(\cdot)$ and $\textup{DOR}(\cdot)$ as follows:

\begin{itemize}
  \item Initially,  $\tilde{\mathcal{E}}_S(\varepsilon)=\textup{DUR}(\emptyset, S(\varepsilon))$;
  \item For all $t^i,t^i\sigma_{i+1} \in \overline{\{t\}}$, $i=0,1,\ldots,|t|-1$,  $$\tilde{\mathcal{E}}_S(t^i\sigma_{i+1})=\textup{DUR}(\textup{DOR}(\tilde{\mathcal{E}}_S(t^i), \sigma_{i+1}),S(t^i\sigma_{i+1})).$$
\end{itemize}

\end{definition}

An example to illustrate the state estimation process will be provided in the next section.
We next discuss the relationship between  $\tilde{\mathcal{E}}_S(t)$ and ${\mathcal{E}}_S(t)$.

\begin{proposition}\label{Prop3}
Given a system $G$ and a networked supervisor $S$ defined over $\Sigma_o^*$, for all $t \in f^{-1}(\psi^f(\mathcal{L}(G_S)))$, we have
\begin{align}\label{Eq13}
&(q,\theta_o,\theta_c)\in \tilde{\mathcal{E}}_S(t) \Rightarrow (\exists \mu \in \mathcal{L}(G_S))\notag \\
& f^{-1}(\psi^f(\mu))=t \wedge \tilde{\delta}(\tilde{q}_0,\mu)=(a,x,n,\phi,y,m,p)\wedge \notag \\
& q=a \wedge \theta_o=(x,n) \wedge \theta_c=(\phi,y,m).
\end{align}
\end{proposition}
\begin{proposition}\label{Prop4}
For any $\mu \in \mathcal{L}(G_S)$, we write $\tilde{\delta}(\tilde{q}_0,\mu)=\tilde{q}=(a,x,n,\phi,y,m,p)$.
Then,  $(q,\theta_c,\theta_o) \in \tilde{\mathcal{E}}_S(f^{-1}(\psi^f(\mu)))$, where $q=a$, $\theta_o=(x,n)$, and $\theta_c=(\phi,y,m)$.
\end{proposition}

Given a set of augmented states $Z \in  2^{Q \times \Theta_o \times \Theta_c}$, let $\textup{FC}(Z)=\{q \in Q:(\exists (q,\theta_o,\theta_c) \in Z)\}$ be the set of first components of augmented states in $Z$.
The following theorem shows that $\textup{FC}(\tilde{\mathcal{E}}_S(t))$ indeed estimates the states of the controlled system.

\begin{theorem}\label{Theo1}
Given automaton $G$ and a networked supervisor $S$, for all $t \in f^{-1}(\psi^f(\mathcal{L}(G_S)))$,  $\textup{FC}(\tilde{\mathcal{E}}_S(t))={\mathcal{E}}_S(t)$.
\end{theorem}
\begin{proof}
We first prove $\textup{FC}(\tilde{\mathcal{E}}_S(t)) \subseteq {\mathcal{E}}_S(t)$.
For any $(q,\theta_o,\theta_c) \in \tilde{\mathcal{E}}_S(t)$, by Proposition \ref{Prop3} and Definition \ref{Def2}, $q \in {\mathcal{E}}_S(t)$.
Therefore,  $\textup{FC}(\tilde{\mathcal{E}}_S(t)) \subseteq {\mathcal{E}}_S(t)$.
Next, we prove that  ${\mathcal{E}}_S(t) \subseteq \textup{FC}(\tilde{\mathcal{E}}_S(t))$.
For any $q \in {\mathcal{E}}_S(t)$, by Definition \ref{Def2}, $\exists \mu \in \mathcal{L}(G_S)$ such that $f^{-1}(\psi^f(\mu))=t$ and $q=\delta(q_0,\psi(\mu))$.
We write $\tilde{\delta}(\tilde{q}_0,\mu)=(a,x,n,\phi,y,m,p)$.
By Proposition \ref{Prop1}, $q=a$.
By Proposition \ref{Prop4}, $(q',\theta_c,\theta_o) \in \tilde{\mathcal{E}}_S(f^{-1}(\psi^f(\mu)))=\tilde{\mathcal{E}}_S(t)$, where $q'=a$, $\theta_o=(x,n)$, and $\theta_c=(\phi,y,m)$.
Therefore, $q=a=q' \in \textup{FC}(\tilde{\mathcal{E}}_S(t))$.
Since $q$ is arbitrarily given,  ${\mathcal{E}}_S(t) \subseteq \textup{FC}(\tilde{\mathcal{E}}_S(t))$.
\end{proof}

\begin{figure}
	\begin{center}
		\includegraphics[width=7.5cm]{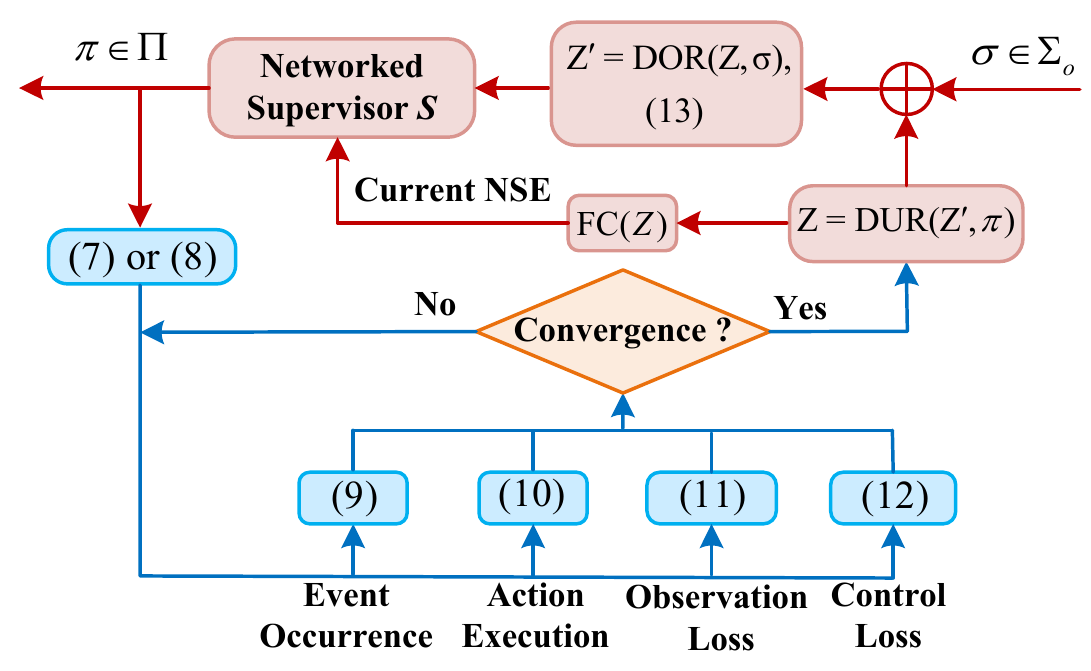}    
		\caption{Online state estimation under communication delays and losses.} 
		\label{Fig4}
	\end{center}
\end{figure}

\begin{remark}
The online  process for estimating the states of the controlled system under communication delays and losses is depicted in Fig.\ref{Fig4}, which is briefly summarized as repeatedly executing: (i) an
observable event occurrence $\sigma \in \Sigma_o$ is communicated to the networked supervisor, and $Z'$ is updated to $Z'=\textup{DOR}(Z,\sigma)$; 
(ii)
a newly issued control action $\pi \in \Pi$ is sent to the actuator of the plant, and the augmented state estimate $Z=\textup{DUR}(Z',\pi)$ is then calculated using (\ref{Eq7-1})$\sim$(\ref{Eq11}).
By Theorem \ref{Theo1}, the current NSE can be obtained by taking all the first components of augmented states in $Z$, i.e., $\textup{FC}(Z)$.
\end{remark}

\begin{remark}
By Fig. \ref{Fig4}, the augmented state estimate is updated only when  a new control action $\pi \in \Pi$ is issued (following a new observation of $\sigma\in \Sigma_o$).
For any $Z \in  2^{Q \times \Theta_o \times \Theta_c}$, the complexities for computing $\textup{DOR}(Z,\sigma)$ and  $\textup{DUR}(\textup{DOR}(Z,\sigma),\pi)$ are linear in the size of $\Theta_o$ and  $Q \times \Theta_o\times \Theta_c$, respectively.
Therefore, the computational complexity for the augmented state estimation is a stepwise order of $\mathcal{O}(|Q| \times |\Theta_o| \times |\Theta_c|)$.
Since $\Theta_o \subseteq (\Sigma_o\times [0,N_o])^{\le N_o+1}\times [0,N_{l,o}]$ and  $\Theta_c \subseteq \Pi \times (\Pi \times [0,N_c])^{\le N+1} \times [0,N_{l,c}]$, we have
\begin{align*}
\mathcal{O}(|Q| \times |\Theta_o| \times |\Theta_c|)=&\mathcal{O}(|Q|\times |\Sigma|^{N_o}\times 2^{N \times |\Sigma|}\times (N_o+1)^{N_o}\\
&  \times (N_c+1)^N \times N_{l,o} \times N_{l,c}).
\end{align*}
\textcolor{blue}{The complexity of the proposed approach after each new observation is polynomial with respect to (w.r.t.) $|Q|$, $N_{l,c}$ and $N_{l,o}$ but is exponential w.r.t. $|\Sigma|$, $N_c$, and $N_o$.
The complexity of the proposed approach grows rapidly with the cardinality of the event set and the delay bounds $N_c$ and $N_o$. 
Thus, the proposed approach is more suitable for estimating the states of a networked DES with relatively small $|\Sigma|$, $N_c$, and $N_o$.
}
\end{remark}

\section{\textcolor{blue}{Comparison with the existing work}}

In this section, we show the difference between the proposed state estimation approach with that proposed in \cite{liu21tac}.
To be consistent with \cite{liu21tac},  we assume in this section that there are only control delays with an upper bound of $N_c$, and there are no control losses, observation delays, and observation losses, i.e., $N_{l,c}=N_o=N_{l,o}=0$.

We first recall the framework adopted in \cite{liu21tac} for specifying the language of the controlled system under control delays.
The framework adopted by \cite{liu21tac} was first proposed in \cite{lin2014control}, where an event $\sigma$ can occur after a string $s \in \mathcal{L}(G)$ if $\sigma$ is allowed to occur by one of the control commands issued in the past $N_{c}$ steps, i.e., $S(P(s_{-i}))$, $i=0,1,\ldots,N_c$.
Formally, the dynamics of the controlled system were specified in \cite{liu21tac} as follows.
\begin{definition}
For a networked DES $G$ with a supervisor $S$, the  language $\mathcal{L}_a(S/G)$ that may be generated by the controlled system is defined recursively as follows:
\begin{align*}
&\varepsilon \in \mathcal{L}_a(S/G);\\
& s\sigma \in \mathcal{L}_a(S/G) \Leftrightarrow s \in \mathcal{L}_a(S/G)\\
& \wedge s\sigma \in \mathcal{L}(G)\wedge(\exists i \in [0,N_c])\sigma \in S(P(s_{-i})).
\end{align*}
\end{definition} 

For all $t \in P(\mathcal{L}_a(S/G))$,  techniques were developed in \cite{liu21tac} for estimating states $q\in Q$ such that there exists $s\in \mathcal{L}_a(S/G)$ with $q=\delta(q_0,s)$ and $t=P(s)$.
Specifically, in \cite{liu21tac}, the state of the control channel is modeled as a set of control actions that have been issued in the past $N_c$ steps, and the state of the controlled system (named as the extended state in \cite{liu21tac}) consists of both the state of the plant and the state of the control channel.
An event can occur at an extended state, if and only if, it is active at the plant state and is allowed to occur by one of the control actions issued in the past $N_c$ steps.
Then, the state estimate of the controlled system can be calculated by estimating all the possible extended states that the controlled system may be in.
Next, we show  that the state estimate calculated in \cite{liu21tac} is actually  an overapproximation of the exact state estimate of the controlled system and may contain states that the controlled system never reaches.

\begin{example}\label{Ex4}
Again, we consider $G$ depicted in Fig. \ref{Fig21} with $\Sigma_c=\Sigma=\{\alpha,\beta,\lambda\}$ and  $\Sigma_o=\{\alpha\}$.
The networked supervisor $S=(A,\gamma)$ is depicted in  Fig. \ref{Fig22}.
We have $S(\varepsilon)=\pi_1=\{\alpha,\lambda\}$, $S(\alpha)=\pi_2=\{\beta\}$, and $S(t)=\pi_3=\emptyset$ for all $t \in \Sigma_o^* \setminus \{\varepsilon,\alpha\}$.
The control delays are upper bounded by 2, i.e., $N_c=2$, and $N_{l,c}=N_o=N_{l,o}=0$.
We first show that $s=\alpha\beta\lambda \in \mathcal{L}_a(S/G)$.

\begin{enumerate}
\item $\varepsilon \in \mathcal{L}_a({S}/G)$;
\item Since $\alpha \in S(\varepsilon)$, we have $\alpha \in \mathcal{L}_a({S}/G)$;
\item Since $\beta \in S(P(\alpha))=S(\alpha)$, $\alpha \in \mathcal{L}_a({S}/G)$, and $\alpha\beta\in \mathcal{L}(G)$, we have $\alpha\beta \in \mathcal{L}_a({S}/G)$;
\item Since $\lambda \in S(P((\alpha\beta)_{-2}))=S(\varepsilon)$, $\alpha\beta \in \mathcal{L}_a({S}/G)$, and  $\alpha\beta\lambda \in \mathcal{L}(G)$, we have $\alpha\beta\lambda \in \mathcal{L}_a({S}/G)$.
\end{enumerate}

Since $s=\alpha\beta\lambda \in \mathcal{L}_a({S}/G)$ and $P(s)=\alpha$ and $q_5=\delta(q_1,s)$, $q_5$ is included in the state estimate calculated in \cite{liu21tac} for $\alpha \in \mathcal{L}_a({S}/G)$.
We next show that $\alpha\beta\lambda$ never occurs in practice.
Since $\lambda \in S(\varepsilon)$ and $\lambda \notin S(\alpha)$, $\lambda$ can occur after $\alpha\beta$ only if the control action that is taking effect after $\alpha\beta$ is  $S(\varepsilon)$.
However, since $\beta \in S(\alpha)$ and $\beta \notin S(\varepsilon)$,  $S(\alpha)$ must have been executed at the time  $\beta$ occurs after $\alpha$.
That is, $S(\varepsilon)$ has been replaced by $S(\alpha)$ after the occurrence of $\alpha\beta$.
Thus, $\alpha\beta\lambda$ never occurs in reality, and the controlled system never reaches $q_5$.
\end{example}

By Example \ref{Ex4},  not all the control actions issued in the past $N_c$ steps can take effect at the moment.
Thus, an event that is active at the plant state and is allowed to occur by one of the control actions issued in the past $N_c$ steps could never occur in practice.
As a result, the state estimate calculated in \cite{liu21tac} may contain states that the controlled system never is in.
In contrast to \cite{liu21tac}, we introduce a new modeling framework for supervisory control of networked DESs (in Section III). 
The state estimation approach proposed under the framework differs crucially from that presented in \cite{liu21tac} in the sense that the proposed approach explicitly models the control action (in the \textsl{control channel configuration}) that can really take effect at each instant. 
As shown in the following example, the proposed approach improves the previous approach because it excludes those states that the controlled system never reaches.

\begin{example}\label{Ex5}
We continue with Example \ref{Ex4}.
We calculate $\tilde{\mathcal{E}}_S(\alpha)$ and ${\mathcal{E}}_S(\alpha)$ using approaches proposed in this paper.

Initially, by Definition \ref{Def5} and  (\ref{Eq7-1}), $\tilde{q}_1=(q_1,(\varepsilon,0),(\pi_1,\varepsilon,0)) \in \tilde{\mathcal{E}}_S(\varepsilon)$.
Since $\delta(q_1,\alpha)!$, $\alpha \in \pi_1$, and $\textbf{NUM}(\varepsilon^+)=0 \le N_c, N_o$, by (\ref{Eq8}), $\tilde{q}_2=(q_2,((\alpha,0),0),(\pi_1,\varepsilon,0)) \in \tilde{\mathcal{E}}_S(\varepsilon)$.
Thus, $\tilde{\mathcal{E}}_S(\varepsilon)=\{\tilde{q}_1,\tilde{q}_2\}$. 

Next, if event $\alpha$ is observed, by $(\ref{Eq12})$, $\textup{DOR}(\tilde{\mathcal{E}}_S(\varepsilon),\alpha)=\{(q_2,(\varepsilon,0),(\pi_1,\varepsilon,0))\}.$
Upon the observation of $\alpha$, $S$ issues $\pi_2$.
By Definition \ref{Def5}, $\tilde{\mathcal{E}}_S(\alpha)=\textup{DUR}(\textup{DOR}(\tilde{\mathcal{E}}_S(\varepsilon),\alpha),\pi_2)$.
By (\ref{Eq7}), $\tilde{q}_3=(q_2,(\varepsilon,0),(\pi_1,(\pi_2,0),0))\in \tilde{\mathcal{E}}_S(\alpha)$.
Since $\delta(q_2,\lambda)=q_3$, $\lambda \in \pi_1$, $\textbf{NUM}(\varepsilon^+)=0 \le N_o$, and $\textbf{NUM}((\pi_2,0)^+)=1 \le N_c$, only event $\lambda$ can occur at $\tilde{q}_3$ ($\beta$ is disabled by $\pi_1$).
By (\ref{Eq9}), $\tilde{q}_4=(q_3,(\varepsilon,0),(\pi_1,(\pi_2,1),0))\in \tilde{\mathcal{E}}_S(\alpha)$.
When control action $\pi_2$ is executed, by (\ref{Eq9}), $\tilde{q}_5=(q_2,(\varepsilon,0),(\pi_2,(\varepsilon,0),0))\in \tilde{\mathcal{E}}_S(\alpha)$ and  $\tilde{q}_6=(q_3,(\varepsilon,0),(\pi_2,(\varepsilon,0),0))\in \tilde{\mathcal{E}}_S(\alpha)$.
Since $\delta(q_3,\beta)=q_4$, $\beta \in \pi_2$, and $\textbf{NUM}(\varepsilon^+)=0 \le N_c, N_o$, by (\ref{Eq8}), we have that $\tilde{q}_7=(q_4,(\varepsilon,0),(\pi_2,\varepsilon,0)) \in \tilde{\mathcal{E}}_S(\alpha)$.
Since $\lambda \notin \pi_2$ and only $\lambda$ is active at $q_4$ in $G$,  $\lambda$ cannot occur at $\tilde{q}_7$. 
Thus, $\alpha\beta\lambda$ will never occur under $S$.
Overall, $\tilde{\mathcal{E}}_S(\alpha)=\{\tilde{q}_3,\tilde{q}_4,\tilde{q}_5,\tilde{q}_6,\tilde{q}_7\}$.
By Theorem \ref{Theo1}, $\mathcal{E}_S(\alpha)=\{q_2,q_3,q_4\}$.
\end{example}

By Example \ref{Ex5}, $\mathcal{E}_S(\alpha)$ does not contain $q_5$. We have shown in Example \ref{Ex4} that the controlled system never reaches $q_5$  under $S$.
The above example justifies the difference and advantage of the proposed approach compared with that proposed in \cite{liu21tac}.

\section{Application}

In this section, we consider the application of the proposed approach.
We first introduce the definition of networked safety.
We then show how to apply the proposed approach to construct an NBTS.
Finally, we discuss how to synthesize a maximally permissible and  networked  safe  supervisor from an NBTS.

\subsection{Networked safety}

We start by defining the networked safety of the DESs under communication delays and losses.
Let $H=(Q_H,\Sigma,\delta_H, q_0) \sqsubseteq G$ be a subautomaton of $G$ that characterizes the specification language (safe behaviors).
That is, $Q_H$ captures all the safe behaviors in the sense that all the strings generated by $G$ are safe if they are ended in states in $Q_H$ and unsafe if they are ended in states in $Q\setminus Q_H$.
Then, the networked safety property of the DESs can be defined as follows.
\begin{definition}\label{Def6}
Given a system automaton $G$, a specification automaton $H$, and a networked supervisor $S$ defined over $\Sigma_o^*$,  we say that $\mathcal{L}(S/G)$ is networked safe w.r.t. $Q_H\subseteq Q$ and $G$ if $(\forall s \in  \mathcal{L}(S/G))\delta(q_0,s)\in Q_H$.
\end{definition}

We next show that networked safety can be formulated as a state-estimate-based (SE-based)  property (or information-state-based property in \cite{yin16tac1,yin16tac2}).
We define the SE-based property $\varphi$ w.r.t. $G$ as a function $\varphi: 2^{Q} \rightarrow \{0,1\}$, where  for all $Z \in 2^{Q}$, $\varphi(Z)=1$ means that $Z$ satisfies property  $\varphi$.

\textcolor{blue}{To ensure that $\mathcal{L}(S/G)$ is networked safe, by Definition \ref{Def6}, we must ensure all the states that the controlled system may reach are within $Q_H$. 
In this regard, a state estimate  $Z\in 2^Q$  is safe if and only if $Z \subseteq Q_H$.
Thus, the definition of the SE-based property $\varphi_{safe}$ is defined as follows. }

\begin{definition}\label{Def7}
The SE-based property $\varphi_{safe}:2^{Q} \rightarrow \{0,1\}$ is defined as follows: for any $Z \in 2^{Q }$,
\begin{align}\label{Eq14}
&\varphi_{safe}(Z)=1 \Leftrightarrow  Z \subseteq Q_H.
\end{align}
\end{definition}


\begin{proposition}\label{Prop5}
\textcolor{blue}{Given automata $G$ and $H$ and a networked supervisor $S$ defined over $\Sigma_o^*$,
$\mathcal{L}(S/G)$ is networked safe  w.r.t. $Q_H\subseteq Q$ and $G$ if and only if all the state estimates that may be generated by the controlled system   satisfy  $\varphi_{safe}$, i.e., $(\forall t \in f^{-1}(\psi^f(\mathcal{L}(G_S))))\varphi_{safe}({\mathcal{E}}_S(t))=1$.}
\end{proposition}
\begin{proof}
($\Rightarrow$) By contradiction. 
Suppose $\exists t\in f^{-1}(\psi^f(\mathcal{L}(G_S)))$ such that $\varphi_{safe}({\mathcal{E}}_S(t))=0$.
By (\ref{Eq14}), $\exists q \in {\mathcal{E}}_S(t)$ such that $q \in Q\setminus Q_H$.
By the definition of ${\mathcal{E}}_S(\cdot)$, $\exists \mu \in \mathcal{L}(G_S)$ such that $f^{-1}(\psi^f(\mu))=t$  and $q=\delta(q_{0},\psi(\mu))$.
Therefore, $\exists \psi(\mu) \in \psi(\mathcal{L}(G_S))$ such that  $\delta(q_0,\psi(\mu))=q \in Q \setminus Q_H$.
By Definition \ref{Def1}, $\exists \psi(\mu) \in \mathcal{L}(S/G)$ such that  $\delta(q_0,\psi(\mu))=q \in Q \setminus Q_H$.
By Definition \ref{Def6}, $\mathcal{L}(S/G)$ is not networked safe.
($\Leftarrow$)
Also by contradiction.
Suppose that $\mathcal{L}(S/G)$ is not networked safe.
By Definitions \ref{Def1} and \ref{Def6}, $\exists \mu \in \mathcal{L}(G_S)$ such that $\delta(q_{0},\psi(\mu)) \in Q\setminus Q_H$.
We write $\tilde{\delta}(\tilde{q}_0,\mu)=(q,x,n,\phi,y,m,p)$ and $f^{-1}(\psi^f(\mu))=t$.
By Proposition \ref{Prop1}, $q=\delta(q_0,\psi(\mu))\in Q\setminus Q_H$.
By the definition of ${\mathcal{E}}_S(\cdot)$,  $q \in {\mathcal{E}}_S(t)$.
Since $q \in Q\setminus Q_H$, $\varphi_{safe}({\mathcal{E}}_S(t))=0$, which contradicts $\varphi_{safe}({\mathcal{E}}_S(t))=1$.
\end{proof}

\textcolor{blue}{Proposition \ref{Prop5} shows that the networked safety enforcement problem can be reduced to a SE-based property $\varphi_{safe}$ enforcement problem.
To ensure that the language of the controlled system is networked safe, it is only required that all the {state estimates} that may be generated by the controlled system satisfy $\varphi_{safe}$.
Using the proposed state estimation algorithm, we next show how to extend a BTS to its network counterpart NBTS, which exhaustively searches all the admissible control actions and state estimates that may be generated  under these control actions.
Benefitting from such a ``global view'', we can always synthesize a supervisor (if it exists) from the NBTS such that all the state estimates that may be generated under it  satisfy $\varphi_{safe}$.
}


\subsection{{Networked supervisor synthesis}}


We first generalize a BTS to an NBTS using the introduced techniques.
Formally, an NBTS $T$ w.r.t. $G$ is a seven-tuple
\begin{align}
T=(Q_Y^T,Q_Z^T,h^T_{YZ},h^T_{ZY},\Sigma_o,\Pi,y_0),
\end{align}
where $Q_Y^T \subseteq Q \times \Theta_o \times \Theta_c \times \Pi$ is the set of $Y$-states; $Q_Z^T \subseteq (Q \times \Theta_o \times \Theta_c \times \Pi)\times \Pi$ is the set of $Z$-states, and each $Z$-state $z=(I(z),\Pi(z))$ consists of two parts  such that $I(z)$ and $\Pi(z)$ denote the information state and the control command parts of $z$, respectively; $h^T_{YZ}:Q_Y^T\times \Pi \rightarrow Q_Z^T$ is a transition function from $Y$-states to $Z$-states, which is defined as follows: for any $y\in Q_Y^T$ and any $\pi \in \Pi$, $$h^T_{YZ}(y,\pi)=(\textup{DUR}(y,\pi),\pi);$$
$h^T_{ZY}:Q_Z^T \times \Sigma_o \rightarrow Q_Y^T$ is a transition function from $Z$-states to $Y$-states, which is defined as follows: for any $z \in Q_Z^T$  and any $\sigma \in \Sigma_o$, $$h^T_{ZY}(z,\sigma)=\textup{DOR}(I(z),\sigma);$$
$\Sigma_o$ is the set of observable events;
$\Pi$ is the set of admissible control commands;
$y_0=\emptyset$ is the initial $Y$-state.

\begin{remark}
The NBTS is an extension of the BTS in the case of communication delays and losses. 
An NBTS also consists of two types of states, named Y-states and Z-states.
A Y-state estimates all the augmented states that the system can  reach immediately after a new observable event communication (by applying ``delayed observable reach'' on its predecessor). 
From a Y-state, all the admissible control decisions are considered. 
A Z-state collects all the augmented states that are reachable from its
predecessor Y-state under a given control action (by applying  ``delayed observable reach'' on its predecessor).
\end{remark}

We say that an NBTS $T$ satisfies the SE-based property $\varphi_{safe}$ if for any  $Z$-state $z \in Q_Z^T$,  all the first components of its information state part satisfy $\varphi_{safe}$, i.e., $\varphi_{safe}(\textup{FC}(I(z)))=1$.
\textcolor{blue}{An NBTS traverses the entire reachable space of the $Y$- and $Z$-states. 
Specifically, in each $Y$-state $y$, the NBTS considers all the admissible control actions and $Z$-states that can be reached from $y$ following the execution of these control actions.
Some of these control actions may be a ``bad'' decision since they may cause the $Z$-states to violate the property of $\varphi_{safe}$ now or in the future.
To exclude all these ``bad'' control actions, we next compute the largest subgraph of an NBTS, called \emph{All Inclusive Networked Controller} (AINC), which  searches only  the ``good'' control decisions.}

We first introduce several notions.
For each $Y$-state $y \in Q_Y^T$, we denote by $C_T(y)=\{\pi \in \Pi:h_{YZ}^T(y,\pi)!\}$ the set of control actions that are defined at $y$.
\begin{definition}\label{Def8}
We say that an NBTS $T$ is \emph{complete}, if 
\begin{enumerate}
\item For all $y \in Q_Y^T$, $C_T(y)\neq \emptyset$;
\item For all $z \in Q_Z^T$ and all $\sigma \in \Sigma_o$, $(\exists (q,\theta_o,\theta_c) \in I(z))[\theta_o]_1=(\sigma_1,n_1)\cdots (\sigma_k,n_k)\neq \varepsilon \wedge \sigma=\sigma_1$ implies $ h^T_{ZY}(z,\sigma)!$.
\end{enumerate}
\end{definition}
The first property says that for any reachable $Y$-state, there exists at least one control action that is defined at this state,
and the second property says that if an observable event is active at a $Z$-state, we cannot disable its occurrence.

Given an NBTS $T$, we denote $Ac(T)$ the accessible part of $T$.
$Ac(T)$ can be obtained by deleting all the $Y$- and $Z$- states in $T$ that are not reachable from the initial $Y$-state $y_0$, and all the transitions that are attached to these states.

With the above preparations, we next construct the AINC.

\begin{definition}\label{Def9}
The AINC is the largest subgraph of an NBTS $T$, denoted by $T'=(Q_Y^{T'},Q_Z^{T'},h^{T'}_{YZ},h^{T'}_{ZY},\Sigma_o,\Pi,y_0),$
such that (i) $T'$ is complete, and (ii) $T'$ satisfies
the defined SE-based property $\varphi_{safe}$, i.e., $(\forall z \in Q_Z^{T'})\varphi_{safe}(\textup{FC}(I(z)))=1$.
\end{definition}

\textcolor{blue}{The AINC is an extension of \textsl{All Inclusive Controller} (AIC) \cite{yin16tac1} in the networked DESs.
It is the largest subgraph of an NBTS $T$ that satisfies completeness and the SE-based property $\varphi_{safe}$.
 We can build an AINC $T'$ in a similar way as the authors  in \cite{yin16tac1} constructed the AIC.
Roughly speaking, the procedure consists of the following two steps: (i) we construct the NBTS $T$ and  prune all its $Z$-states $z\in Q_Z^T$ that violate $\varphi_{safe}$, i.e., $\varphi_{safe}(\textup{FC}(I(z)))=0$; (ii) we repeatedly prune all the states violating completeness from the remaining part of $T$ until convergence is achieved.
Algorithm 1 formally constructs $T'$.}

\begin{algorithm}[htb]
	\LinesNumbered
		\KwIn{Automaton ${G}$ and SE-based property $\varphi_{safe}$}
		\KwOut{An AINC $T'$} 
         We \label{lin1} first construct an NBTS $T$ using $G$ as described above;\\
        We remove \label{lin2} all the $Z$-states $z\in Q_Z^T$ in $T$ such that $\varphi_{safe}(\textup{FC}(I(z)))=0$ from $T$, and set ${T}\leftarrow Ac(T)$;\\
        \Repeat {${Q}_Y^T=\tilde{Q}_Y^T$ and ${Q}_Z^T=\tilde{Q}_Z^T$}
{
Set $\tilde{Q}_Y^T \leftarrow {Q}_Y^T$ and $\tilde{Q}_Z^T \leftarrow {Q}_Z^T$;\\
\For{$y \in {Q}_Y^T$, one by one \label{lin10}}
{
\If{$C_T(y)=\emptyset$\label{lin12}}
{Set \label{lin13}${Q}_Y^T \leftarrow {Q}_Y^T \setminus \{y\}$;}
Set \label{lin13} ${T}\leftarrow Ac(T)$;
}
\For{$z \in {Q}_Z^T$, one by one \label{lin10}}
{
\If{$(\exists (q,\theta_o,\theta_c) \in I(z))[\theta_o]_1=(\sigma_1,n_1)\cdots (\sigma_k,n_k)\neq \varepsilon \wedge h^T_{ZY}(z,\sigma_1)\not !$\label{lin12}}
{Set \label{lin13}${Q}_Z^T \leftarrow {Q}_Z^T \setminus \{z\}$;\\
}
Set \label{lin13} ${T}\leftarrow Ac(T)$;\\
}
}
\Return{$T' \leftarrow T$.}
\caption{\textcolor{blue}{\textsc{Calculating $T'$}}}
	\end{algorithm}

\begin{proposition}\label{Prop6}
Algorithm 1 correctly constructs the AINC $T'$.
\end{proposition}

Let $T'$ be the returned AINC of Algorithm 1.
We write $y \xrightarrow{\pi} z$ if $h^{T'}_{YZ}(y,\pi)=z$ and $z \xrightarrow{\sigma} y$ if $h^{T'}_{ZY}(z,\sigma)=y$.
Let $S$ be a networked supervisor included in $T'$ and $t=\sigma_1\cdots\sigma_n \in f^{-1}(\psi^f(\mathcal{L}(G_S)))$ be an observed string.
The execution of $t$ leads to  an alternating sequence of $Y$-states and $Z$-states
$$y_0 \xrightarrow{S(\varepsilon)} z_0 \xrightarrow{\sigma_1} y_1 \xrightarrow{S(\sigma_1)}\cdots  \xrightarrow{\sigma_n} y_n \xrightarrow{S(\sigma_1\cdots\sigma_n)}z_n.$$
We denote $IS_{S}^Y(t)$ and $IS_{S}^Z(t)$ as the last $Y$- and $Z$-states of $y_0z_0\cdots y_nz_n$, respectively, i.e., $IS_{S}^Y(t)=y_n$ and $IS_{S}^Z(t)=z_n$.
We now show how to ``decode'' a networked supervisor from $T'$.

\begin{definition}\label{Def10}
A networked supervisor $S=(A,\gamma)$ with $A=(X,\Sigma_o,\xi, x_0)$ is said to be included in an AINC $T'$, if for all $t \in f^{-1}(\psi^f(\mathcal{L}(G_S)))$,  $\xi(x_0,t)=I(IS_{S}^{Z}(t))$ and $\gamma(\xi(x_0,t))\in C_{T'}(IS_{S}^{Y}(t))$.
To complete the definition of $S$, for all $t \in \Sigma_o^* \setminus f^{-1}(\psi^f(\mathcal{L}(G_S)))$,  define $\xi(x_0,t)=x_{spec}$ with $\gamma(x_{spec})=\Sigma_{uc}$.
\end{definition}

We denote by $\mathbb{S}(T')$ all the networked supervisors included in $T'$.
{
The following theorem and its corollary show that $\mathbb{S}(T')$ collects only networked safe supervisors.}
\begin{theorem}\label{Theo2}
Let $S \in \mathbb{S}(T')$ be a networked supervisor included in $T'$.
For all $t \in f^{-1}(\psi^f(\mathcal{L}(G_S)))$,  ${\mathcal{E}}_S(t)=\textup{FC}(I(IS_{S}^{Z}(t)))$. 
\end{theorem}

Then, we have the following corollary of Theorem \ref{Theo2}.
\begin{corollary}\label{Coro1}
Let $S \in \mathbb{S}(T')$ be a networked supervisor included in $T'$. 
Then, $\mathcal{L}(S/G)$ is networked safe  w.r.t. $Q_H\subseteq Q$ and $G$.
\end{corollary}
\begin{proof}
By Theorem \ref{Theo2} and the fact that $T'$ is an AINC, we have $\varphi_{safe}({\mathcal{E}}_S(t))=1$ for all $t \in f^{-1}(\psi^f(\mathcal{L}(G_S)))$.
By Proposition \ref{Prop5}, $\mathcal{L}(S/G)$ is networked safe. 
\end{proof}

\textcolor{blue}{
By Corollary \ref{Coro1}, $\mathbb{S}(T')$ collects only networked safe supervisors.
We can always select a maximal networked safe supervisor from $\mathbb{S}(T')$ in the following sense: Let $IS^Y_{S}(t) \in Q_Y^{T'}$ be the current $Y$-state such that AINC $T'$ is in (after observing $t$).
Since $\mathbb{S}(T')$ collects only networked safe supervisors, all the control actions in $C_{T'}(IS^Y_{S}(t))$ are safe control actions. 
We can simply pick a ``greedy maximal'' control action from $C_{T'}(IS^Y_{S}(t))$.
Note that there may be several incomparable ``greedy maximal'' control actions in $C_{T'}(IS^Y_{S}(t))$, but all of them are safe and maximal.
Since we focus on estimating states in this paper, the formal algorithm for synthesizing a maximal supervisor is beyond the scope of this paper.}

\section{Conclusion}

In supervisory control, communication delays and losses are unavoidable when
 communication between the plant and the supervisor for
observation, and between the supervisor and the actuator for
control are carried out over some shared networks. We assume, in this paper, that (i) the delays do not change the order of the observations and controls,
(ii) both the observation delays and control delays have upper bounds, and (iii) both the consecutive observation losses and the consecutive control losses also have upper bounds. 
A novel framework for supervisory control under communication delays and losses has been established.
Under this framework, an algorithm for online state estimation of a controlled system has been proposed.
The proposed algorithm can be used to solve the \textsl{supervisor synthesis problem} in networked DESs.
As an application, we show how to use the existing methods to synthesize  maximally permissible and safe networked supervisors.
\bibliographystyle{IEEEtran} 
\bibliography{articles-network}

\appendix

\subsection{Proof of Proposition \ref{Prop1}}
\begin{proof}
The proof is by induction on the length of strings in $\mathcal{L}(G_S)$.
Since $\tilde{\delta}(\tilde{q}_0,\varepsilon)=(q_0,\varepsilon,0,S(\varepsilon),\varepsilon,0,x_0)$, $\delta(q_0,\varepsilon)=q_0$, and  $\xi(x_0,\varepsilon)=x_0$, the base case is true.
The induction hypothesis is that for any $\mu \in \mathcal{L}(G_S)$ with $|\mu| \le n$, if  $\tilde{\delta}(\tilde{q}_0,\mu)=(q,x,n,\phi,y,m,p)$, then $q=\delta(q_0,\psi(\mu))$ and  $p=\xi(x_0,f^{-1}(\psi^f(\mu)))$.
We next prove the same is also true for $\mu e \in \mathcal{L}(G_S)$.
We write $\tilde{\delta}(\tilde{q}_0,\mu e)=(q',x',n',\phi',y',m',p')$.
By the definition of $\tilde{\Sigma}$,  $e \in \Sigma$, $e\in \Sigma^f$, or $e \in \tilde{\Sigma}\setminus (\Sigma \cup \Sigma^f)$.
We consider each of them separately as follows.

\emph{Case 1:} $e \in \Sigma$. 
By (\ref{Eq1}), $q'=\delta(q,e)$ and $p'=p$.
Since $e \in \Sigma$, by definitions, we have $\psi(\mu e)=\psi(\mu)e$ and $f^{-1}(\psi^f(\mu e))=f^{-1}(\psi^f(\mu))$.
Moreover, since $q=\delta(q_0,\psi(\mu))$ and  $p=\xi(x_0,f^{-1}(\psi^f(\mu)))$, we have $\delta(q_0,\psi(\mu e))=\delta(q,e)=q'$ and  $\xi(x_0,f^{-1}(\psi^f(\mu e)))=\xi(x_0,f^{-1}(\psi^f(\mu)))=p=p'$.

\emph{Case 2:} $e=f(\sigma) \in \Sigma^f$. 
By (\ref{Eq3}), $q'=q$ and $p'=\xi(p,\sigma)$.
Since $e \in \Sigma^f$, we have $\psi(\mu e)=\psi(\mu)$ and $f^{-1}(\psi^f(\mu e))=f^{-1}(\psi^f(\mu)) \sigma$.
Moreover, since $q=\delta(q_0,\psi(\mu))$ and  $p=\xi(x_0,f^{-1}(\psi^f(\mu)))$, we have $\delta(q_0,\psi(\mu e))=\delta(q_0,\psi(\mu))=q=q'$ and  $\xi(x_0,f^{-1}(\psi^f(\mu e)))=\xi(x_0,f^{-1}(\psi^f(\mu))\sigma)=\xi(p,\sigma)=p'$.

\emph{Case 3:} $e \in \tilde{\Sigma} \setminus (\Sigma \cup \Sigma^f)$. 
By (\ref{Eq2}), (\ref{Eq4}), and (\ref{Eq55}), $q'=q$ and $p'=p$.
Since $e \in \tilde{\Sigma} \setminus (\Sigma \cup \Sigma^f)$,  $\psi(\mu e)=\psi(\mu)$ and $f^{-1}(\psi^f(\mu e))=f^{-1}(\psi^f(\mu))$.
Moreover, since $q=\delta(q_0,\psi(\mu))$ and  $p=\xi(x_0,f^{-1}(\psi^f(\mu)))$,  $\delta(q_0,\psi(\mu e))=\delta(q_0,\psi(\mu))=q=q'$ and  $\xi(x_0,f^{-1}(\psi^f(\mu e)))=\xi(x_0,f^{-1}(\psi^f(\mu)))=p=p'$.
\end{proof}

\subsection{Proof of Proposition \ref{Prop2}}

\begin{proof}
For any $y=(\pi_1,l_1)\cdots (\pi_h,l_h) \in (\Pi \times [0,N_c])^{\le N+1}$ and any $y'=(\pi'_1,l'_1)\cdots (\pi'_h,l'_h) \in (\Pi \times [0,N_c])^{\le N+1}$, we say that $y$ is smaller than $y'$, denoted by $y \preceq y'$,  if $\pi_i \subseteq \pi'_i$ and $l_i=l_i'$ for all $i=1,\ldots,h$.
Note that $\varepsilon \preceq \varepsilon$ always holds.
For a $\mu_1 \in \mathcal{L}(G_{S_1})$, let $\tilde{\delta}_1(\tilde{q}_{0,1},\mu_1)=(q_1,x_1,n_1,\phi_1,y_1,m_1,p_1)$.
We next prove there always exists a string $\mu_2 \in \mathcal{L}(G_{S_2})$ such that $\psi(\mu_1)=\psi(\mu_2)$, $\psi^f(\mu_1)=\psi^f(\mu_2)$, and $\tilde{\delta}_2(\tilde{q}_{0,2},\mu_2)=(q_2,x_2,n_2,\phi_2,y_2,m_2,p_2)$ with $q_2=q_1$, $x_2=x_1$, $n_2=n_1$, $\phi_1 \subseteq \phi_2$, $y_1 \preceq y_2$, and $m_2=m_1$.
The proof is by induction on the length of strings in $\mathcal{L}(G_{S_1})$.

\textbf{Base case:}
We have  $\tilde{\delta}_1(\tilde{q}_{0,1},\varepsilon)=(q_0,\varepsilon,0,S_1(\varepsilon),\varepsilon,0,x_{0,1})$ and $\tilde{\delta}_2(\tilde{q}_{0,2},\varepsilon)=(q_0,\varepsilon,0,S_2(\varepsilon),\varepsilon,0,x_{0,2}).$
Since $\psi(\varepsilon)=\psi(\varepsilon)$, $\psi^f(\varepsilon)=\psi^f(\varepsilon)$, $q_0=q_0$, $\varepsilon=\varepsilon$, $0=0$, $S_1(\varepsilon)\subseteq S_2(\varepsilon)$, $\varepsilon \preceq \varepsilon$, and $0=0$, the base case is true.

\textbf{Induction hypothesis:} For all $\mu_1 \in \mathcal{L}(G_{S_1})$ with $|\mu_1|\le k$, if $\tilde{\delta}_1(\tilde{q}_{0,1},\mu_1)=(q_1,x_1,n_1,\phi_1,y_1,m_1,p_{1})$, then there exists a  $\mu_2 \in \mathcal{L}(G_{S_2})$ such that $\psi(\mu_1)=\psi(\mu_2)$, $\psi^f(\mu_1)=\psi^f(\mu_2)$, and $\tilde{\delta}_2(\tilde{q}_{0,2},\mu_2)=(q_2,x_2,n_2,\phi_2,y_2,m_2,p_2)$ with $q_2=q_1$, $x_2=x_1$, $n_2=n_1$, $\phi_1\subseteq \phi_2$, $y_1 \preceq y_2$, and $m_2=m_1$.
We next prove the same is also true for $\mu_1e \in \mathcal{L}(G_{S_1})$.
Write $\tilde{\delta}_1(\tilde{q}_{0,1},\mu_1 e)=(q'_1,x'_1,n'_1,\phi'_1,y'_1,m'_1,p'_{1})$.

\emph{Case 1:} $e \in \Sigma$. 
Since $\mu_1e \in \mathcal{L}(G_{S_1})$, by (\ref{Eq1}),  $\delta(q_1,e)!$, $e \in \phi_1$, $\textbf{NUM}(x_1^+) \le N_o$, and $\textbf{NUM}(y_1^+) \le N_c$.
Since $q_1=q_2$, $\phi_1 \subseteq \phi_2$, $x_1=x_2$, and $y_1 \preceq y_2$, we have $\delta(q_2,e)!$, $e \in \phi_2$, $\textbf{NUM}(x_2^+) \le N_o$, and $\textbf{NUM}(y_2^+) \le N_c$.
By (\ref{Eq1}), $\mu_2e \in \mathcal{L}(G_{S_2})$.
We write $\tilde{\delta}_2(\tilde{q}_{0,2},\mu_2 e)=(q'_2,x'_2,n'_2,\phi'_2,y'_2,m'_2,p'_2)$.
By (\ref{Eq1}), $q'_2=\delta(q_2,e)$, $x'_2=x_2^+(e,0)$ if $e \in \Sigma_o$ and $x'_2=x_2^+$ if $e \in \Sigma_{uo}$, $n'_2=n_2$, $\phi_2'=\phi_2$, $y'_2=y_2^+$, and $m_2'=m_2$.
Moreover, since  $\tilde{\delta}_1(\tilde{q}_{0,1},\mu_1 e)=(q'_1,x'_1,n'_1,\phi'_1,y'_1,m'_1,p'_{1})$,
by (\ref{Eq1}), $q'_1=\delta(q_1,e)$, $x'_1=x_1^+(e,0)$ if $e \in \Sigma_o$ and $x'_1=x_1^+$ if $e \in \Sigma_{uo}$, $n'_1=n_1$, $\phi_1'=\phi_1$, $y'_1=y_1^+$, and $m_1'=m_1$.
Since $q_2=q_1$, $q'_2=\delta(q_2,e)$, and $q'_1=\delta(q_1,e)$, we have $q'_2=q'_1$.
Since $x_2=x_1$, $x'_2=x_2^+(e,0)$ and $x'_1=x_1^+(e,0)$ if $e \in \Sigma_o$,   $x'_2=x_2^+$ and $x'_1=x_1^+$ if $e \in \Sigma_{uo}$, we have $x'_2=x'_1$.
Meanwhile, by the induction hypothesis, we have  $n_2=n_1$, $\phi_1\subseteq \phi_2$, $y_1 \preceq y_2$, and $m_2=m_1$.
Since $n'_1=n_1$, $\phi_1'=\phi_1$, $y'_1=y_1^+$, $m_1'=m_1$, $n'_2=n_2$, $\phi_2'=\phi_2$, $y'_2=y_2^+$, and $m_2'=m_2$, we have $n'_2=n'_1$, $\phi'_1 \subseteq \phi'_2$, $y'_1 \preceq y'_2$, and $m'_2=m'_1$.
Overall, we have $q'_1=q'_2$, $x'_1=x'_2$, $n'_1=n'_2$, $\phi_1'\subseteq \phi_2'$, $y'_1 \preceq y'_2$, and $m_1'=m_2'$.
Since  $\psi(\mu_1)=\psi(\mu_2)$ and $\psi^f(\mu_1)=\psi^f(\mu_2)$, we have $\psi(\mu_1 e)=\psi(\mu_1) e=\psi(\mu_2) e=\psi(\mu_2 e)$ and $\psi^f(\mu_1 e)=\psi^f(\mu_1)=\psi^f(\mu_2)=\psi^f(\mu_2 e)$.

\emph{Case 2:} $e=o(i) \in \Sigma^o$. 
Since $\mu_1o(i) \in \mathcal{L}(G_{S_1})$, by (\ref{Eq2}), $x_1=(\sigma_1,n_1)\cdots (\sigma_k,n_k) \neq\varepsilon$, $i\in [1,k]$, and $n_1+1\le N_{l,o}$.
Since $x_1=x_2$ and $n_1=n_2$ and  $\tilde{\delta}_2(\tilde{q}_{0,2},\mu_2)=(q_2,x_2,n_2,\phi_2,y_2,m_2,p_2)$, we have $\mu_2o(i) \in \mathcal{L}(G_{S_2})$.
Let us write $\tilde{\delta}_2(\tilde{q}_{0,2},\mu_2 e)=(q'_2,x'_2,n'_2,\phi'_2,y'_2,m'_2,p'_2)$.
By (\ref{Eq2}), $q'_2=q_2$, $$x'_2=(\sigma_1,n_1) \cdots (\sigma_{i-1},n_{i-1})(\sigma_{i+1},n_{i+1})\cdots (\sigma_k,n_k),$$  $n'_2=n_2+1$, $\phi_2'=\phi_2$, $y'_2=y_2$, and $m_2'=m_2$.
Moreover, since  $\tilde{\delta}_1(\tilde{q}_{0,1},\mu_1 e)=(q'_1,x'_1,n'_1,\phi'_1,y'_1,m'_1,p'_{1})$,
by (\ref{Eq2}), $q'_1=q_1$, $$x'_1=(\sigma_1,n_1)\cdots (\sigma_{i-1},n_{i-1})(\sigma_{i+1},n_{i+1})\cdots (\sigma_k,n_k),$$ $n'_1=n_1+1$, $\phi_1'=\phi_1$, $y'_1=y_1$, and $m_1'=m_1$.
Hence, we have $x'_1=x'_2$.
Since $n_1=n_2$, $n'_1=n_1+1$, and $n'_2=n_2+1$, we have $n'_1=n'_2$.
Moreover, by the induction hypothesis, $q_2=q_1$,  $\phi_1\subseteq \phi_2$, $y_1\preceq y_2$, and $m_2=m_1$.
Sicne $q'_1=q_1$, $\phi_1'=\phi_1$, $y'_1=y_1$, $m_1'=m_1$,  $q'_2=q_2$, $\phi_2'=\phi_2$, $y'_2=y_2$, and $m_2'=m_2$, we have $q'_2=q'_1$, $\phi'_1 \subseteq \phi'_2$, $y'_1 \preceq y'_2$, and  $m'_2=m'_1$.
Overall, we have $q'_1=q'_2$, $x'_1=x'_2$, $n'_1=n'_2$, $\phi_1'\subseteq \phi_2'$, $y'_1 \preceq y'_2$, and $m_1'=m_2'$.
Since  $\psi(\mu_1)=\psi(\mu_2)$ and $\psi^f(\mu_1)=\psi^f(\mu_2)$, by definitions, $\psi(\mu_1 e)=\psi(\mu_1)=\psi(\mu_2)=\psi(\mu_2 e)$ and $\psi^f(\mu_1 e)=\psi^f(\mu_1)=\psi^f(\mu_2)=\psi^f(\mu_2 e)$.

\emph{Case 3:} $e=f(\sigma) \in \Sigma^f$.  
Since $\mu_1f(\sigma) \in \mathcal{L}(G_{S_1})$, by (\ref{Eq3}),  $x_1=(\sigma_1,n_1)\cdots (\sigma_k,n_k) \neq \varepsilon$ and $\sigma=\sigma_1$.
Since $x_1=x_2$, by (\ref{Eq3}),  $\mu_2f(\sigma) \in \mathcal{L}(G_{S_2})$.
We write $$\tilde{\delta}_2(\tilde{q}_{0,2},\mu_2 e)=(q'_2,x'_2,n'_2,\phi'_2,y'_2,m'_2,p'_2).$$
By (\ref{Eq3}), $q'_2=q_2$, $x'_2=(\sigma_2,n_2) \cdots (\sigma_k,n_k)$, $n'_2=0$, $\phi_2'=\phi_2$, $y'_2=y_2(\gamma_2(\xi_2(p_2,\sigma)),0)$, and $m_2'=m_2$.
Let us write $t=f^{-1}(\psi^f(\mu_2))$.
By Proposition \ref{Prop1}, $\xi_2(x_{0,2},t)=p_2$. 
Thus, $\gamma_2(\xi_2(p_2,\sigma))=S_2(t\sigma)$ and  $y'_2=y_2(S_2(t\sigma),0)$.
Furthermore, since  $\tilde{\delta}_1(\tilde{q}_{0,1},\mu_1 e)=(q'_1,x'_1,n'_1,\phi'_1,y'_1,m'_1, p'_{1})$,
by (\ref{Eq3}), $q'_1=q_1$, $x'_1=(\sigma_2,n_2) \cdots (\sigma_k,n_k)$, $n'_1=0$, $\phi_1'=\phi_1$, $y'_1=y_1(\gamma_1(\xi_1(p_1,\sigma)),0)$, and $m_1'=m_1$.
Therefore,  $x'_1=x_2'$ and $n'_1=n_2'=0$.
Meanwhile, since $f^{-1}(\psi^f(\mu_1))=f^{-1}(\psi^f(\mu_2))$,  $f^{-1}(\psi^f(\mu_1))=t$.
By Proposition \ref{Prop1}, $p_1=\xi_1(x_{0,1},t)$.
Hence, $\gamma_1(\xi_1(p_1,\sigma))=S_1(t\sigma)$ and $y'_1=y_1(S_1(t\sigma),0)$.
Sicne $y'_1=y_1(S_1(t\sigma),0)$, $y'_2=y_2(S_2(t\sigma),0)$, $y_1 \preceq y_2$, and $S_1(t\sigma) \subseteq  S_2(t\sigma)$, we have $y_1'\preceq y_2'$.
Meanwhile, by the induction hypothesis, $q_2=q_1$, $\phi_1 \subseteq \phi_2$, and $m_2=m_1$.
Since $q'_1=q_1$, $\phi_1'=\phi_1$, $m_1'=m_1$, $q'_2=q_2$, $\phi_2'=\phi_2$, and $m_2'=m_2$, we have $q'_2=q'_1$, $\phi'_1 \subseteq \phi'_2$, and $m'_2=m'_1$.
Overall, we have $q'_2=q'_1$, $x'_2=x'_1$, $n'_2=n'_1$, $\phi'_1\subseteq \phi'_2$, $y'_1 \preceq y'_2$, and $m'_2=m'_1$.
Since  $\psi(\mu_1)=\psi(\mu_2)$ and $\psi^f(\mu_1)=\psi^f(\mu_2)$, by definitions, $\psi(\mu_1 e)=\psi(\mu_1)=\psi(\mu_2)=\psi(\mu_2 e)$ and $\psi^f(\mu_1 e)=\psi^f(\mu_1)e=\psi^f(\mu_2)e=\psi^f(\mu_2 e)$.

\emph{Case 4:}  $e=c(i) \in \Sigma^c$. 
Since $\mu_1c(i) \in \mathcal{L}(G_{S_1})$, by (\ref{Eq4}), $y_1=(\pi_1,l_1)\cdots (\pi_h,l_h) \neq\varepsilon$, $i \in [1,h]$, and $m_1+1\le N_{l,c}$.
Since $y_1 \preceq y_2$, by definition, $y_2=(\pi'_1,l'_1) \cdots (\pi'_h, l'_h)$ such that $\pi_i \subseteq \pi_i'$ and $l_i=l_i'$ for all $i \in [1,h]$.
Then, by (\ref{Eq4}), we have $\mu_2c(i) \in \mathcal{L}(G_{S_2})$.
We write $\tilde{\delta}_2(\tilde{q}_{0,2},\mu_2 e)=(q'_2,x'_2,n'_2,\phi'_2,y'_2,m'_2,p'_2)$.
By (\ref{Eq4}), $q'_2=q_2$, $x'_2=x_2$, $n'_2=n_2$, $\phi_2'=\phi_2$, $m_2'=m_2+1$, and $$y'_2=(\pi'_1,l'_1) \cdots (\pi'_{i-1},l'_{i-1})(\pi'_{i+1},l'_{i+1})\cdots (\pi'_h,l'_h).$$
Since  $\tilde{\delta}_1(\tilde{q}_{0,1},\mu_1 e)=(q'_1,x'_1,n'_1,\phi'_1,y'_1,m'_1,p'_{1})$,
by (\ref{Eq4}),  $q'_1=q_1$, $x'_1=x_1$, $n'_1=n_1$, $\phi_1'=\phi_1$,  $m_1'=m_1+1$, and $$y'_1=(\pi_1,l_1) \cdots (\pi_{i-1},l_{i-1})(\pi_{i+1},l_{i+1})\cdots (\pi_h,l_h).$$
Since $\pi_i \subseteq \pi_i'$ and $l_i=l_i'$ for all $i \in [1,h]$,  $y'_1 \preceq y'_2$.
Since $m_1=m_2$, $m'_1=m_1+1$, and $m'_2=m_2+1$,  $m'_1=m'_2$.
Moreover, by the induction hypothesis, $q_2=q_1$, $x_2=x_1$, $n_2=n_1$, and $\phi_1\subseteq \phi_2$.
Sicne $q'_1=q_1$, $x'_1=x_1$, $n_1'=n_1$, $\phi_1'=\phi_1$, $q'_2=q_2$, $x'_2=x_2$, $n_2'=n_2$, and $\phi_2'=\phi_2$, we have $q'_2=q'_1$, $x'_2=x'_1$, $n_2'=n_1'$, and $\phi'_1 \subseteq \phi'_2$.
Overall, we have $q'_2=q'_1$, $x'_2=x'_1$, $n'_2=n'_1$, $\phi'_1\subseteq \phi'_2$, $y'_1 \preceq y'_2$, and $m'_2=m'_1$.
Since  $\psi(\mu_1)=\psi(\mu_2)$ and $\psi^f(\mu_1)=\psi^f(\mu_2)$, by definitions, $\psi(\mu_1 e)=\psi(\mu_1)=\psi(\mu_2)=\psi(\mu_2 e)$ and $\psi^f(\mu_1 e)=\psi^f(\mu_1)=\psi^f(\mu_2)=\psi^f(\mu_2 e)$.

\emph{Case 5:} $e=g(\pi) \in \Sigma^g$. 
Since $\mu_1g(\pi) \in \mathcal{L}(G_{S_1})$, by (\ref{Eq55}), $y_1 \neq \varepsilon$.
Let us write $y_1=(\pi_1,l_1)\cdots (\pi_h,l_h)$ for $\pi_i\in \Pi$ and $l_i \le N_c$.
Then, $\pi=\pi_1$.
Since $y_1 \preceq y_2$,  $y_1=(\pi'_1,l'_1)\cdots (\pi'_h,l'_h)$ with $\pi_i \subseteq \pi_i'$ and $l_i=l_i'$.
By (\ref{Eq55}), $\mu_2e' \in \mathcal{L}(G_{S_2})$ with $e'=g(\pi_1')$.
We write $\tilde{\delta}_2(\tilde{q}_{0,2},\mu_2 e')=(q'_2,x'_2,n'_2,\phi'_2,y'_2,m'_2,p'_2)$.
By (\ref{Eq55}), $q'_2=q_2$, $x'_2=x_2$, $n'_2=n_2$, $\phi_2'=\pi_1'$, $y'_2=(\pi'_2,l'_2)\cdots (\pi'_h,l'_h)$, and $m_2'=0$.
Furthermore, since  $\tilde{\delta}_1(\tilde{q}_{0,1},\mu_1 e)=(q'_1,x'_1,n'_1,\phi'_1,y'_1, m'_1, p'_{1})$,
by (\ref{Eq55}), we have $q'_1=q_1$, $x'_1=x_1$, $n'_1=n_1$, $\phi_1'=\pi_1$, $y'_1=(\pi_2,l_2)\cdots (\pi_h, l_h)$, and $m_1'=0$.
Since $\pi_1\subseteq \pi_2$, we have $\phi_1'\subseteq \phi_2'$.
Since  $y'_1=(\pi_2,l_2)\cdots (\pi_h, l_h)$ and $y'_2=(\pi'_2,l'_2)\cdots (\pi'_h,l'_h)$, we have $y'_1 \preceq y'_2$.
Since $m_1'=0$ and $m_2'=0$, we know $m_1'=m_2'$.
Meanwhile, by the induction hypothesis, $q_2=q_1$, $x_2=x_1$, and  $n_2=n_1$.
Since $q'_1=q_1$, $x'_1=x_1$, $n'_1=n_1$, $q'_2=q_2$, $x'_2=x_2$, and $n'_2=n_2$, we have $q'_1=q'_2$, $x'_1=x'_2$, and $n'_1=n'_2$.
Overall, we have $q'_1=q'_2$, $x'_1=x'_2$, $n'_1=n'_2$, $\phi_1'\subseteq \phi_2'$, $y'_1 \preceq y'_2$, and $m_1'=m_2'$.
Since  $\psi(\mu_1)=\psi(\mu_2)$ and $\psi^f(\mu_1)=\psi^f(\mu_2)$, by definitions, $\psi(\mu_1 e)=\psi(\mu_1)=\psi(\mu_2)=\psi(\mu_2 e')$ and $\psi^f(\mu_1 e)=\psi^f(\mu_1)=\psi^f(\mu_2)=\psi^f(\mu_2 e')$.

By the above proof, for all $\mu_1 \in \mathcal{L}(G_{S_1})$, $\exists \mu_2 \in \mathcal{L}(G_{S_2})$ such that $\psi(\mu_1)=\psi(\mu_2)$.
Hence, $\psi(\mathcal{L}(G_{S_1}))\subseteq \psi(\mathcal{L}(G_{S_2}))$.
By Definition \ref{Def1}, $\mathcal{L}(S_1/G) \subseteq \mathcal{L}(S_2/G)$.
\end{proof}

\subsection{Proof of Proposition \ref{Prop3}}

\begin{proof}
Let us first define
\begin{align*}
& \mathcal{T}(t)=\{(q,\theta_o,\theta_c): (\exists \mu \in \mathcal{L}(G_S))f^{-1}(\psi^f(\mu))=t \wedge \tilde{\delta}(\tilde{q}_0,\mu)\\
& =(a,x,n,\phi,y,m,p)\wedge q=a \wedge \theta_o=(x,n) \wedge \theta_c=(\phi,y,m)\}.
\end{align*}
Next, we prove $\tilde{\mathcal{E}}_S(t)\subseteq \mathcal{T}(t)$ by induction on the length of  $t \in f^{-1}(\psi^f(\mathcal{L}(G_S)))$.

\textbf{Base case:} For any $z=(q,\theta_c,\theta_o) \in \tilde{\mathcal{E}}_S(\varepsilon)$, by the definition of $\textup{DUR}(\cdot)$,  there exists a sequence of augmented states $z_0z_1 \cdots z_k$ such that $z_0=(q_0,(\varepsilon,0),(S(\varepsilon),\varepsilon,0))$, $z_k=z$, and $z_i$ is the augmented state calculated by applying one of operations in (\ref{Eq8})$\sim$(\ref{Eq11}) on $z_{i-1}$, $i=1,\ldots,k$.
Since the initial state of $G_S$ is $z_0=(q_0, \varepsilon, 0, S(\varepsilon), \varepsilon, 0, x_0)$, by the definition of $\mathcal{T}(\cdot)$, $z_0 \in {\mathcal{T}}(\varepsilon)$.
Next, we introduce the following claim.
\begin{claim}\label{Cla1}
For any $t \in f^{-1}(\psi^f(\mathcal{L}(G_S)))$, if $z \in \mathcal{T}(t)$ and $z'$ is the augmented state obtained by applying one of operations in (\ref{Eq8})$\sim$(\ref{Eq11}) on $z$, then $z' \in \mathcal{T}(t)$.
\end{claim}
We prove Claim \ref{Cla1} in Appendix D.
Since $z_0 \in {\mathcal{T}}(\varepsilon)$, by recursively applying Claim \ref{Cla1},  $z=z_k \in {\mathcal{T}}(\varepsilon)$.
Therefore, the base case holds, i.e., $\tilde{\mathcal{E}}_S(\varepsilon) \subseteq \mathcal{T}(\varepsilon)$.

\textbf{Induction hypothesis:} For all $t \in f^{-1}(\psi^f(\mathcal{L}(G_S)))$ with $|t|\le n$, we have $\tilde{\mathcal{E}}_S(t)\subseteq \mathcal{T}(t)$.
We next prove the same is also true for $t\sigma \in f^{-1}(\psi^f(\mathcal{L}(G_S)))$.

If $\tilde{\mathcal{E}}_S(t\sigma)=\emptyset$, then $\tilde{\mathcal{E}}_S(t\sigma)\subseteq \mathcal{T}(t\sigma)$ is trivially true.
Otherwise, for any $z \in \tilde{\mathcal{E}}_S(t\sigma)$, by the definitions of $\textup{DOR}(\cdot)$ and $\textup{DUR}(\cdot)$, there exists a sequence of augmented states $z_0z_1 \cdots z_k$ such that (i) $z_0$ is the augmented state calculated by applying successively (\ref{Eq12}) and (\ref{Eq7}) on a $(q,\theta_o,\theta_c) \in \tilde{\mathcal{E}}_S(t)$, i.e., $z_0=(q, \textbf{OUT}^{obs}(\theta_o),\textbf{IN}^{ctr}(\theta_c,S(t\sigma)))$, (ii) $z_k=z$, and (iii) $z_i$ is the augmented state calculated by applying one of operations in (\ref{Eq8})$\sim$(\ref{Eq11}) on $z_{i-1}$, $i=1,\ldots,k$.

Next, we prove $z_0 \in {\mathcal{T}}(t\sigma)$.
Since $(q,\theta_o,\theta_c) \in \tilde{\mathcal{E}}_S(t) \subseteq \mathcal{T}(t)$, there exists $\mu \in \mathcal{L}(G_S)$ with $f^{-1}(\psi^f(\mu))=t$ and $\tilde{\delta}(\tilde{q}_0,\mu)=(a,x,n,\phi,y,m,p)$ such that $q=a$, $\theta_o=(x,n)$, and $\theta_c=(\phi,y,m)$.
Since $\textbf{OUT}^{obs}(\theta_o)$ is defined, we have $x \neq \varepsilon$.
We write $x=(\sigma_1,n_1)\cdots (\sigma_k,n_k)$ for $\sigma_i\in \Sigma_o$ and $n_i \in [0,N_o]$. 
Then, by (\ref{Eq3}), $\tilde{\delta}(\tilde{q}_0,\mu f(\sigma))=(a',x',n',\phi',y',m',p')$, where $a'=a$, $x'=(\sigma_2,n_2)\cdots (\sigma_k,n_k)$,  $n'=0$, $\phi'=\phi$, $y'=y(\gamma(\xi(p,\sigma)),0)$, $m'=m$, and $p'=\xi(p,\sigma)$.
Since $f^{-1}(\psi^f(\mu f(\sigma)))=t\sigma$, by the definition of $\mathcal{T}(\cdot)$,  $(a',(x',n'),(\phi',y',m'))\in \mathcal{T}(t\sigma)$.
Since $f^{-1}(\psi^f(\mu))=t$, by Proposition \ref{Prop1},  $p=\xi(x_0,t)$.
Hence, $\xi(p,\sigma)=\xi(x_0,t\sigma)$.
By definition, $\gamma(\xi(p,\sigma))=S(t\sigma)$.
Therefore, $y'=y(\gamma(\xi(p,\sigma)),0)=y(S(t\sigma),0)$. 
Moreover, since $a'=q$, $\phi'=\phi$, $n'=0$, and $m'=m$, we have $$(a',(x',n'),(\phi',y',m'))=(q,(x',0),(\phi,y(S(t\sigma),0),m))\in \mathcal{T}(t\sigma).$$

Since $\theta_o=(x,n)$, by the definition of $\textbf{OUT}^{obs}(\cdot)$,  $\theta'_o=(x',0)$.
Since $\theta_c=(\phi,y,m)$, by the definition of $\textbf{IN}^{ctr}(\cdot)$, $\textbf{IN}^{ctr}(\theta_c,S(t\sigma))=(\phi,y(S(t\sigma),0),m)$.
Therefore, 
\begin{align*}
z_0&=(q, \textbf{OUT}^{obs}(\theta_o),\textbf{IN}^{ctr}(\theta_c,S(t\sigma)))\\
&=(q,(x',0),(\phi,y(S(t\sigma),0),m))\in \mathcal{T}(t\sigma).
\end{align*}
Since $z_0 \in {\mathcal{T}}(t\sigma)$, by recursively applying  Claim \ref{Cla1},  $z=z_k \in {\mathcal{T}}(t\sigma)$.
Therefore, $\tilde{\mathcal{E}}_S(t\sigma) \subseteq \mathcal{T}(t\sigma)$.
That completes the proof.
\end{proof}

\subsection{Proof of Claim \ref{Cla1}}

\begin{proof}

We write $z=(q,\theta_{o},\theta_{c})$ and $z'=(q',\theta'_{o},\theta'_{c})$.
Since $z \in \mathcal{T}(t)$, there exists a $\mu \in \mathcal{L}(G_S)$ such that $f^{-1}(\psi^f(\mu))=t$ and $\tilde{\delta}(\tilde{q}_0,\mu)=\tilde{q}=(a,x,n,\phi,y,m,p)$ with $q=a$, $\theta_{o}=(x,n)$, and $\theta_{c}=(\phi,y,m)$.
Since $z'$ is the augmented state obtained by applying one of operations in (\ref{Eq8})$\sim$(\ref{Eq11}) on $z$, one of the following four cases must be true.

\emph{{Case 1:}} 
$z'=(\delta(q,\sigma), \textbf{IN}^{obs}(\theta_{o},\sigma),\textbf{PLUS}(\theta_{c}))$.
By (\ref{Eq8}), there exists a $\sigma \in \Sigma$ with $\delta(q,\sigma)!$, $\sigma \in [\theta_{c}]_1$, $\textbf{NUM}([\theta_{o}]_1^+) \le N_o$, $\textbf{NUM}([\theta_{c}]_2^+) \le N_c$.  
Since $q=a$, $[\theta_{c}]_1=\phi$, $[\theta_{o}]_1=x$, and $[\theta_{c}]_2=y$, we have $\delta(a,\sigma)!$,  $\sigma \in \phi$, $\textbf{NUM}(x^+)\le N_o$, and $\textbf{NUM}(y^+) \le N_c$.
By (\ref{Eq1}), $\sigma$ is defined at $\tilde{q}$ in $G_S$.
We write $\tilde{\delta}(\tilde{q},\sigma)=\tilde{\delta}(\tilde{q}_0,\mu\sigma)=(a',x',n',\phi',y',m',p')$.
By (\ref{Eq1}), we have  $a'=\delta(a,\sigma)$, $x'=x^+(\sigma,0)$ if $\sigma \in \Sigma_o$, $x'=x^+$ if $\sigma \in \Sigma_{uo}$,  $n'=n$, $\phi'=\phi$,  $y'=y^+$, $m'=m$, and $p'=p$.
Since $\sigma \in \Sigma$, we have $f^{-1}(\psi^f(\mu\sigma))=f^{-1}(\psi^f(\mu))=t$.
By the definition of $\mathcal{T}(\cdot)$, we know if $\sigma \in \Sigma_o$, $(a',(x', n'),(\pi', y',m'))=(\delta(a,\sigma),(x^+(\sigma,0),n),(\phi, y^+,m))\in \mathcal{T}(t)$, and if $\sigma \in \Sigma_{uo}$, $(a',(x', n'),(\pi', y',m'))=(\delta(a,\sigma),(x^+,n),(\phi, y^+,m))\in \mathcal{T}(t)$.
By the definition of $\textbf{IN}^{obs}(\cdot)$,  if $\sigma \in \Sigma_o$, $\textbf{IN}^{obs}(\theta_{o},\sigma)=(x^+(\sigma,0),n)$, and if $\sigma \in \Sigma_{uo}$, $\textbf{IN}^{obs}(\theta_{o},\sigma)=(x^+,n)$.
By the definition of $\textbf{PLUS}(\cdot)$, $\textbf{PLUS}(\theta_{c})=(\phi, y^+,m)$.
Moreover, since $q=a$,  $z'=(\delta(q,\sigma), \textbf{IN}^{obs}(\theta_{o},\sigma),\textbf{PLUS}(\theta_{c}))\in \mathcal{T}(t)$.

\emph{{Case 2:}}  
$z'=(q, \theta_{o},\textbf{OUT}^{ctr}(\theta_{c}))$.
By (\ref{Eq9}), $[\theta_{c}]_2=y \neq \varepsilon$.
We write $y=(\pi_1,m_1)\cdots (\pi_h,m_h)$ for $\pi_i \in \Pi$ and $m_i \le N_c$.
By (\ref{Eq55}), $g(\pi_1)$ is defined at $\tilde{q}$ in $G_S$.
We write $\tilde{\delta}(\tilde{q},g(\pi_1))=\tilde{\delta}(\tilde{q}_0,\mu\sigma)=(a',x',n',\phi',y',m',p')$.
By (\ref{Eq55}), $a'=a$, $x'=x$, $n'=n$, $\phi'=\pi_1$, $y'=(\pi_2,m_2)\cdots (\pi_h,m_h)$, $m'=0$, and $p'=p$.
Since $g(\pi_1) \in \Sigma_g$, $f^{-1}(\psi^f(\mu g(\pi_1)))=f^{-1}(\psi^f(\mu))=t$.
By the definition of $\mathcal{T}(\cdot)$, we have $(a',(x', n'),(\pi', y',m'))=(a,(x, n),(\pi_1, y',0))\in \mathcal{T}(t)$.
Sicne $\theta_c=(\phi,y,m)$, by the definition of $\textbf{OUT}^{ctr}(\cdot)$,  $\textbf{OUT}^{ctr}(\theta_{c})=(\pi_1,y',0)$.
Moreover, since $q=a$ and $\theta_o=(x,n)$, we have $z'=(q, \theta_{o},\textbf{OUT}^{ctr}(\theta_{c}))=(q,(x, n),(\pi_1, y',0)) \in \mathcal{T}(t)$.

\emph{{Case 3:}} 
$z'=(q, \textbf{LOSS}^{obs}(\theta_{o},i),\theta_{c})$.
By (\ref{Eq10}), $[\theta_{o}]_1=x \neq \varepsilon$, $[\theta_{o}]_2 +1=n+1 \le N_{l,o}$, and $i \in [1, |x|]$. 
We write $x=(\sigma_1,n_1)\cdots (\sigma_k,n_k)$ for $\sigma_i\in \Sigma_o$ and $n_i \le N_o$.
Since $i \in [1, |x|]=[1,k]$, by (\ref{Eq2}), $o(i)$ is defined at $\tilde{q}$ in $G_S$.
Let us write $\tilde{\delta}(\tilde{q},o(i))=\tilde{\delta}(\tilde{q}_0,\mu o(i))=(a',x',n',\phi',y',m',p')$.
By (\ref{Eq2}), we have $a'=a$, $x'=(\sigma_1,n_1)\cdots (\sigma_{i-1},n_{i-1}) (\sigma_{i+1},n_{i+1})\cdots (\sigma_k,n_k),$  $n'=n+1$, $\phi'=\phi$, $y'=y$, $m'=m$, and $p'=p$. 
Since $o(i) \in \Sigma^o$, we have  $f^{-1}(\psi^f(\mu o(i)))=f^{-1}(\psi^f(\mu))=t$.
By the definition of $\mathcal{T}(\cdot)$, $(a',(x', n'), (\phi', y', m'))=(a,(x', n+1), (\phi, y, m))\in \mathcal{T}(t)$.
Since $\theta_{o}=(x,n)$, by the definition of $\textbf{LOSS}^{obs}(\cdot)$, $\textbf{LOSS}^{obs}(\theta_{o},i)=(x',n+1)$.
Moreover, since $q=a$ and $\theta_c=(\phi, y, m)$, we have $z'=(q, \textbf{LOSS}^{obs}(\theta_{o},i),\theta_{c})=(a,(x', n+1), (\phi, y, m)) \in \mathcal{T}(t)$.

\emph{{Case 4:}}
$z'=(q, \theta_{o}, \textbf{LOSS}^{ctr}(\theta_{c},i))$.
By (\ref{Eq11}),  $[\theta_{c}]_2=y  \neq \varepsilon$, $[\theta_{c}]_3+1=m+1 \le N_{l,c}$, and $i \in [1, |y|]$. 
We write $y=(\pi_1,m_1)\cdots (\pi_h,m_h)$ for $\pi_i \in \Pi$ and $m_i \le N_c$.
Since $i \in [1,|y|]=[1,h]$, by (\ref{Eq4}),  $c(i)$ is defined at $\tilde{q}$ in $G_S$.
We write $\tilde{\delta}(\tilde{q},c(i))=\tilde{\delta}(\tilde{q}_0,\mu c(i))=(a',x',n',\phi',y',m',p')$.
By (\ref{Eq4}), we have $a'=a$, $x'=x$, $n'=n$, $\phi'=\phi$, $m'=m+1$, $p'=p$, and $y'=(\pi_1,m_1)\cdots (\pi_{i-1},m_{i-1}) (\pi_{i+1},m_{i+1})\cdots (\pi_h,m_h)$. 
Since $c(i) \in \Sigma^c$,   $f^{-1}(\psi^f(\mu c(i)))=f^{-1}(\psi^f(\mu))=t$.
By the definition of $\mathcal{T}(\cdot)$,  $(a',(x', n'),(\phi', y', m'))=(a,(x,n),(\phi,y',m+1))\in \mathcal{T}(t)$.
By the definition of $\textbf{LOSS}^{ctr}(\cdot)$, $\textbf{LOSS}^{ctr}(\theta_{c},i)=(\phi,y',m+1)$.
Hence,  $\textbf{LOSS}^{ctr}(\theta_{o},i)=(\phi,y',m+1)$.
Moreover, since $q=a$ and $\theta_o=(x,n)$, it has $z'=(q, \theta_{o}, \textbf{LOSS}^{ctr}(\theta_{c},i))=(a,(x,n),(\phi,y',m+1)) \in \mathcal{T}(t)$.
\end{proof}

\subsection{Proof of Proposition \ref{Prop4}}
\begin{proof}
The proof is by induction on the length of $\mu \in \mathcal{L}(G_S)$.

\textbf{Base case:} By definition, $\tilde{\delta}(\tilde{q}_0,\varepsilon)=(q_0,\varepsilon, 0, S(\varepsilon),\varepsilon, 0, x_0)$.
Since $\tilde{\mathcal{E}}_S(\varepsilon)=\text{DUR}(\emptyset,S(\varepsilon))$, by (\ref{Eq7-1}), $(q_0,(\varepsilon,0),(S(\varepsilon),\varepsilon,0)) \in \tilde{\mathcal{E}}_S(\varepsilon)$.
The base case is true.

\textbf{Induction hypothesis:} For all $\mu \in \mathcal{L}(G_S)$ with $|\mu| \le k$, we write  $\tilde{\delta}(\tilde{q}_0,\mu)=\tilde{q}=(a,x,n,\phi,y,m,p)$.
We also write $f^{-1}(\psi^f(\mu))=t$.
Then, $(q,\theta_c,\theta_o) \in \tilde{\mathcal{E}}_S(f^{-1}(\psi^f(\mu)))=\tilde{\mathcal{E}}_S(t)$, where $q=a$, $\theta_o=(x,n)$, and $\theta_c=(\phi,y,m)$.
We next prove the same is also true for $\mu e \in \mathcal{L}(G_S)$.

For any $\mu e \in \mathcal{L}(G_S)$ with $|\mu| = k$, let $\tilde{\delta}(\tilde{q}_0,\mu e)=\tilde{q}'=(a',x',n',\phi',y',m',p')$.
Since $\tilde{\delta}(\tilde{q}_0,\mu)=\tilde{q}$, we have $\tilde{\delta}(\tilde{q},e)=\tilde{q}'$.
Let $q'=a'$, $\theta'_o=(x',n')$, and $\theta'_c=(\phi',y',m')$.
We next prove $(q',\theta'_o,\theta'_c) \in \tilde{\mathcal{E}}_S(f^{-1}(\psi^f(\mu e)))$.
Since $e\in \tilde{\Sigma}$, one of the following cases must be true: (i) $e\in \Sigma$, (ii) $e \in \Sigma^o$, (iii) $e \in \Sigma^c$, (iv) $e \in \Sigma^g$, and (v) $e \in \Sigma^f$.
We consider each of them separately as follows.

\emph{Case 1:} $e \in \Sigma$. Since $e \in \Sigma$, $f^{-1}(\psi^f(\mu))=f^{-1}(\psi^f(\mu e))=t$.
Thus, to prove $(q',\theta'_o,\theta'_c)\in \tilde{\mathcal{E}}_S(f^{-1}(\psi^f(\mu e)))$, we only need to prove $(q',\theta'_o,\theta'_c)\in \tilde{\mathcal{E}}_S(t)$ as follows.

Since $\tilde{\delta}(\tilde{q},e)=\tilde{q}'$,  by (\ref{Eq1}),  we have the following two results: (i) $\delta(a,e)!$, $e \in \phi$, $\textbf{NUM}(x^+)\le N_o$, and $\textbf{NUM}(y^+) \le N_c$; and (ii) $a'=\delta(a,e)$,  $x'=x^+$ if $e \in \Sigma_{uo}$ and $x'=x^+(e,0)$ if $e \in \Sigma_{o}$, $n'=n$, $\phi'=\phi$, $y'=y^+$, and $m'=m$.
Since $\theta_o=(x,n)$ and $\theta_c=(\phi,y,m)$, we have $\delta(a,e)!$, $e \in [\theta_c]_1$, $\textbf{NUM}([\theta_o]_1^+) \le N_o$, and $\textbf{NUM}([\theta_c]_2^+) \le N_c$. 
Moreover, since $(q=a,\theta_o,\theta_c) \in \tilde{\mathcal{E}}_S(t)$, by (\ref{Eq8}), $(\delta(a,e), \textbf{IN}^{obs}(\theta_o,e),  \textbf{PLUS}(\theta_c)) \in \tilde{\mathcal{E}}_S(t)$.
We next prove $(q',\theta'_o,\theta'_c)=(\delta(a,e), \textbf{IN}^{obs}(\theta_o,e),  \textbf{PLUS}(\theta_c))$ by considering the cases of $e \in \Sigma_{uo}$ and $e \in \Sigma_o$ separately as follows.
If $e \in \Sigma_{uo}$, $(q',\theta'_o,\theta'_c)=(\delta(a,e), (x^+,n),(\phi, y^+, m))$, and otherwise if $e \in \Sigma_{o}$, $(q',\theta'_o,\theta'_c)=(\delta(a,e), (x^+(e,0),n),(\phi, y^+, m))$.
Since $\theta_o=(x,n)$ and $\theta_c=(\phi,y,m)$, if  $e \in \Sigma_{uo}$, $\textbf{IN}^{obs}(\theta_o,e)=(x^+,n)$ and $\textbf{PLUS}(\theta_c)=(\phi,y^+,m)$, and if $e \in \Sigma_{o}$, $\textbf{IN}^{obs}(\theta_o,e)=(x^+(e,0),n)$ and $\textbf{PLUS}(\theta_c)=(\phi,y^+,m)$.
Thus, $(q',\theta'_o,\theta'_c)=(\delta(a,e), \textbf{IN}^{obs}(\theta_o,e),  \textbf{PLUS}(\theta_c))\in \tilde{\mathcal{E}}_S(t)$.


\emph{Case 2:} $e=o(i) \in \Sigma^o$. 
Since $\tilde{\delta}(\tilde{q},o(i))=\tilde{q}'$, by (\ref{Eq2}),  $x \neq\varepsilon$, $i\in [1,|x|]$, and $n+1 \le N_{l,o}$.
Write $x=(\sigma_1,n_1)\cdots (\sigma_k,n_k)$ for $\sigma_i\in \Sigma_o$ and $n_i \in [0,N_o]$.
Since $\tilde{q}'=(a',x',n',\phi',y',m',p')$, by (\ref{Eq2}), we have $a'=a$, $x'=(\sigma_1,n_1)\cdots (\sigma_{i-1},n_{i-1})(\sigma_{i+1},n_{i+1})\cdots (\sigma_k,n_k),$ $n'=n+1$, $\phi'=\phi$, $y'=y$, and $m'=m$.
Thus, we have $(q',\theta'_o,\theta'_c)=(a',(x',n'),(\phi', y', m'))=(a, (x',n+1),(\phi, y, m))$.
Since $\theta_o=(x,n)$, $x \neq \varepsilon$, $i \in [1,|x|]$, and $n+1\le N_{l,o}$, we have $\textbf{LOSS}^{obs}(\theta_o,i)=(x',n+1)$. 
Moreover, since $\theta_c=(\phi,y,m)$, $(q',\theta'_o,\theta'_c)=(a, (x',n+1),\theta_c)=(a,\textbf{LOSS}^{obs}(\theta_o,i), \theta_c)$.
Since $(q=a,\theta_o,\theta_c) \in \tilde{\mathcal{E}}_S(t)$, by (\ref{Eq10}),  $(a,\textbf{LOSS}^{obs}(\theta_o,i),  \theta_c) \in \tilde{\mathcal{E}}_S(t)$.
Thus, $(q',\theta'_o,\theta'_c) \in \tilde{\mathcal{E}}_S(t)$.
Since $e \in \Sigma^o$, $f^{-1}(\psi^f(\mu e))=f^{-1}(\psi^f(\mu))=t$. 
Thus,  $(q',\theta'_o,\theta'_c) \in \tilde{\mathcal{E}}_S(f^{-1}(\psi^f(\mu e)))$.

\emph{{Case 3:}} $e=f(\sigma) \in \Sigma^f$. 
Since $\tilde{\delta}(\tilde{q},f(\sigma))=\tilde{q}'$, by (\ref{Eq3}),  $x \neq \varepsilon$.
We write $x=(\sigma_1,n_1)\cdots (\sigma_k,n_k)$ for $\sigma_i\in \Sigma_o$ and $n_i \in [0,N_o]$.
Since $\tilde{q}'=(a',x',n',\phi',y',m',p')$, by (\ref{Eq3}),  $a'=a$, $x'=(\sigma_2,n_2) \cdots(\sigma_k,n_k)$, $n'=0$, $\phi'=\phi$, $y'=y(\gamma(\xi(p,\sigma)),0)$, and $m'=m$.
Thus, $(q',\theta'_o,\theta'_c)=(a',(x',n'),(\phi',y',m'))=(a, (x',0),(\phi, y', m))$.
Since $\tilde{\delta}(\tilde{q}_0,\mu)=\tilde{q}=(a,x,n,\phi,y,m,p)$, by Proposition \ref{Prop1}, $p=\xi(x_0,f^{-1}(\psi^f(\mu)))=\xi(x_0,t)$.
Since $p'=\xi(p,\sigma)$, we have $p'=\xi(x_0, t\sigma)$.
By the definition of $S$, we have $\gamma(p')=S(t\sigma)$, which implies $y'=y(S(t\sigma),0)$.
Since $\theta_o=(x,n)$, $\textbf{OUT}^{obs}(\theta_o)=(x', 0)$.
Meanwhile, since $\theta_c=(\phi,y,m)$, by the definition of $\textbf{IN}^{ctr}(\cdot)$, $\textbf{IN}^{ctr}(\theta_c, S(t\sigma))=(\phi, y(S(t\sigma),0), m)$.
Therefore, $(q',\theta'_o,\theta'_c)=(a, \textbf{OUT}^{obs}(\theta_o),\textbf{IN}^{ctr}(\theta_c,S(t\sigma)))$.
Since $(q=a,\theta_o,\theta_c) \in \tilde{\mathcal{E}}_S(t)$, by (\ref{Eq7}) and  (\ref{Eq12}), $(a, \textbf{OUT}^{obs}(\theta_o),\textbf{IN}^{ctr}(\theta_c,S(t\sigma)))\in \textup{DUR}(\textup{DOR}(\tilde{\mathcal{E}}_S(t),\sigma),S(t\sigma))$.
By Definition \ref{Def5}, we have $(a, \textbf{OUT}^{obs}(\theta_o),\textbf{IN}^{ctr}(\theta_c,S(t\sigma))) \in \tilde{\mathcal{E}}_S(t\sigma)$.
Thus, $(q',\theta'_o,\theta'_c) \in \tilde{\mathcal{E}}_S(t\sigma)$.
Since $e=f(\sigma) \in \Sigma^f$, we have  $f^{-1}(\psi^f(\mu e))=t\sigma$. 
Therefore, $(q',\theta'_o,\theta'_c) \in \tilde{\mathcal{E}}_S(f^{-1}(\psi^f(\mu e)))$.

\emph{{Case 4:}} $e=c(i) \in \Sigma^c$. 
Since $\tilde{\delta}(\tilde{q},c(i))=\tilde{q}'$, by (\ref{Eq4}), $y \neq\varepsilon$, $i \in [1,|y|]$, and $m+1 \le N_{l,c}$.
We write $y=(\pi_1,m_1)\cdots (\pi_h,m_h)$ for $\pi_i\in \Pi$ and $m_i \in [0,N_c]$.
Since $\tilde{q}'=(a',x',n',\phi',y',m',p')$, by (\ref{Eq4}),  $a'=a$, $x'=x$, $n'=n$, $\phi'=\phi$, $y'=(\pi_1,m_1)\cdots (\pi_{i-1},m_{i-1})(\pi_{i+1},m_{i+1}) \cdots (\pi_h,m_h),$ and $m'=m+1$.
Therefore, $(q',\theta'_o,\theta'_c)=(a',(x',n'),(\phi', y', m'))=(q, (x,n),(\phi, y', m+1))$.
Since $a'=a$ and $\theta_o=(x,n)$, we have $(q',\theta'_o,\theta'_c)=(a, \theta_o,(\phi, y', m+1))$.
Since $y \neq \varepsilon$, $i\in [1,|y|]$, and $m+1 \le N_{l,c}$, we have $\textbf{LOSS}^{ctr}(\theta_c,i)=(\phi, y', m+1)$.
Thus, $(q',\theta'_o,\theta'_c)=(a, \theta_o,(\phi, y', m+1))=(a, \theta_o, \textbf{LOSS}^{ctr}(\theta_{c},i))$.
Since $(q=a,\theta_o,\theta_c) \in \tilde{\mathcal{E}}_S(t)$, by (\ref{Eq11}),  $(a, \theta_o, \textbf{LOSS}^{ctr}(\theta_{c},i)) \in \tilde{\mathcal{E}}_S(t)$.
Thus,  $(q',\theta'_o,\theta'_c) \in \tilde{\mathcal{E}}_S(t)$.
Moreover, since $e \in \Sigma^c$, $f^{-1}(\psi^f(\mu e))=f^{-1}(\psi^f(\mu))=t$. 
Thus, $(q',\theta'_o,\theta'_c) \in \tilde{\mathcal{E}}_S(f^{-1}(\psi^f(\mu e)))$.

\emph{{Case 5:}} $e=g(\pi) \in \Sigma^g$. 
Since $\tilde{\delta}(\tilde{q},g(\pi))=\tilde{q}'$, by (\ref{Eq55}),  $y \neq\varepsilon$. 
We write $y=(\pi_1,m_1)\cdots (\pi_h,m_h)$ for $\pi_i\in \Pi$ and $m_i \in [0,N_c]$.
Since $\tilde{q}'=(a',x',n',\phi',y',m',p')$, by (\ref{Eq55}), $a'=a$, $x'=x$, $n'=n$, $\phi'=\pi_1$, $y'=(\pi_2,m_2)\cdots (\pi_h,m_h),$ and $m'=m$.
Thus, $$(q',\theta'_o,\theta'_c)=(a', (x',n'),(\pi', y', m'))=(a, (x,n),(\pi_1, y', m)).$$
Since $\theta_o=(x,n)$,  $(q',\theta'_o,\theta'_c)=(a, \theta_o,(\phi, y', m))$.
Meanwhile, since $y\neq\varepsilon$, $\textbf{OUT}^{ctr}(\theta_c)=(\pi_1, y', m)$.
Therefore, $(q',\theta'_o,\theta'_c)=(a, \theta_o, \textbf{OUT}^{ctr}(\theta_{c}))$.
Since $(q=a,\theta_o,\theta_c) \in \tilde{\mathcal{E}}_S(t)$, by (\ref{Eq9}), $(a, \theta_o, \textbf{OUT}^{ctr}(\theta_{c})) \in \tilde{\mathcal{E}}_S(t)$.
Thus, $(q',\theta'_o,\theta'_c) \in \tilde{\mathcal{E}}_S(t)$.
Moreover, since $e \in \Sigma^g$,  $f^{-1}(\psi^f(\mu e))=f^{-1}(\psi^f(\mu))=t$. 
Therefore, $(q',\theta'_o,\theta'_c) \in \tilde{\mathcal{E}}_S(f^{-1}(\psi^f(\mu e)))$.
\end{proof}

\subsection{Proof of Proposition \ref{Prop6}}

\begin{proof}
Let $T$ be the NBTS obtained by executing Lines 1 and 2.
Then, $T$ satisfies $\varphi_{safe}$ because all states that violate $\varphi_{safe}$ have been removed
by Line 2. 
Additionally, by the repeat-until loop on Line 3,  $T'$ is complete and satisfies $\varphi_{safe}$ .
Next, we show that $T'$ is the largest NBTS with the desired properties. 
The proof is by contradiction.
Assume that $T''$ is another NBTS satisfying $\varphi_{safe}$ and completeness that is strictly larger than $T'$. 
Since $T''$ satisfies $\varphi_{safe}$, $T''$ is a subgraph of $T$.
For any $Y$- and $Z$-states in $T''$ satisfying conditions 1) and 2) in Definition \ref{Def8} should also
satisfy these conditions in $T$.
In other words, all the $Y$- and $Z$-states in $T''$ will not be removed from $T$ after executing the repeat-until loop on Line 3.
Therefore, Algorithm 1 will converge to an NBTS that is strictly lager than $T'$ (at least as
large as $T''$). 
This contradicts the fact that Algorithm 1 converges to $T'$.
\end{proof}

\subsection{Proof of Theorem \ref{Theo2}}

\begin{proof}
We first prove that $\tilde{\mathcal{E}}_S(t)=I(IS_{S}^{Z}(t))$ for all $t \in f^{-1}(\psi^f(\mathcal{L}(G_S)))$.
The proof is by induction on the finite length of sequence $t \in f^{-1}(\psi^f(\mathcal{L}(G_S)))$.
By Definition \ref{Def5}, $\tilde{\mathcal{E}}_S(\varepsilon)=\textup{DUR}(\emptyset,S(\varepsilon))$.
Since $y_0=\emptyset$ and $S\in \mathbb{S}(T')$, we have $S(\varepsilon) \in C_{T'}(y_0)$.
By definition, $IS_{S}^{Z}(\varepsilon)=h^{T'}_{YZ}(y_0,S(\varepsilon))=(\textup{DUR}(y_0,S(\varepsilon)),S(\varepsilon))$.
Thus,  $\tilde{\mathcal{E}}_S(\varepsilon)=I(IS_{S}^{Z}(\varepsilon))$. 
The base case is true.
The induction hypothesis is that $\forall t \in f^{-1}(\psi^f(\mathcal{L}(G_S)))$ with $|t| \le n$,  $\tilde{\mathcal{E}}_S(t)=I(IS_{S}^{Z}(t))$. 
We now prove that the same is also true for $t\sigma \in f^{-1}(\psi^f(\mathcal{L}(G_S)))$ with $|t|=n$.
By Definition \ref{Def5}, $\tilde{\mathcal{E}}_S(t\sigma)=\textup{DUR}(\textup{DOR}(\tilde{\mathcal{E}}_S(t),\sigma),S(t\sigma))$.
Since $\textup{DOR}(\tilde{\mathcal{E}}_S(t),\sigma)\neq \emptyset$, there exists $(q,\theta_o,\theta_c) \in \tilde{\mathcal{E}}_S(t)$ such that $[\theta_o]_1=(\sigma_1,n_1)\cdots (\sigma_k,n_k)\neq \varepsilon$ and $\sigma=\sigma_1$.
By the induction hypothesis, $\tilde{\mathcal{E}}_S(t)=I(IS_{S}^{Z}(t))$. 
Moreover, since $T'$ is complete,  $h^{T'}_{ZY}(IS_{S}^{Z}(t),\sigma)=\textup{DOR}(I(IS_{S}^{Z}(t)),\sigma)=\textup{DOR}(\tilde{\mathcal{E}}_S(t),\sigma)=IS_{S}^{Y}(t\sigma)$.
By Definition \ref{Def10}, $S(t\sigma) \in C_{T'}(IS_{S}^{Y}(t\sigma))$.
Thus,
\begin{align*}
IS_{S}^{Z}(t\sigma)&=h^{T'}_{YZ}(IS_{S}^{Y}(t\sigma),S(t\sigma))\\
&=(\textup{DUR}(\textup{DOR}(\tilde{\mathcal{E}}_S(t),\sigma),S(t\sigma)),S(t\sigma)).
\end{align*}
Thus, $\tilde{\mathcal{E}}_S(t\sigma)=I(IS_{S}^{Z}(t\sigma))$.

Therefore, for all $t \in f^{-1}(\psi^f(\mathcal{L}(G_S)))$, $\tilde{\mathcal{E}}_S(t)=I(IS_{S}^{Z}(t))$.
By Theorem \ref{Theo1}, ${\mathcal{E}}_S(t)=\textup{FC}(\tilde{\mathcal{E}}_S(t))=\textup{FC}(I(IS_{S}^{Z}(t)))$.
\end{proof}

\end{document}